\documentclass[conference]{IEEEtran}

\usepackage{cite}
\usepackage{amsmath,amsthm,amssymb,amsfonts}
\usepackage{graphicx}
\usepackage{epstopdf}
\usepackage{textcomp}
\usepackage{xcolor}
\usepackage{color}
\usepackage{setspace}
\usepackage{mathtools}
\usepackage{bm}
\usepackage{booktabs}

\usepackage{mathrsfs}
\usepackage{bm}
\usepackage{pgfplots}
\usepackage{tikz}
\usetikzlibrary{arrows}
\usepackage{subfigure}
\usepackage{graphicx,booktabs,multirow}
\usepackage{upgreek}
\usepackage{bm}
\usepackage{mathrsfs}
\usepackage{subfigure}
\usepackage{algorithmic}
\allowdisplaybreaks[1]

\usepackage[lined,boxed,commentsnumbered, ruled]{algorithm2e}

\definecolor{colorhkust}{RGB}{20,43,140}
\definecolor{colortsinghua}{RGB}{116,52,129}
\definecolor{color1}{RGB}{128,0,0}

\newtheorem{lemma}{Lemma}
\newtheorem{theorem}{Theorem}

\newtheorem{definition}{Definition}
\newtheorem{remark}{Remark}

\newcommand{\mini}{\operatorname{minimize}}

\newcommand{\subj}{\operatorname{subject~to}}

\newcommand{\diagg}{\rm{diag}}
\newcommand{\supp}{{\rm{supp}}}

\newcommand{\R}{\mathbb{R}}  % The real numbers.
\newcommand{\C}{\mathbb{C}} % The complex numbers.
\newcommand{\E}{\mathbb{E}}

\DeclareMathOperator*{\argmin}{argmin}
\IEEEoverridecommandlockouts
\begin{document}

\title{Algorithm Unrolling for Massive Access via Deep Neural Network with Theoretical Guarantee}
%\thanks{This work was supported in part by the National Nature Science Foundation of China under Grant 61601290.}}

%\and
%\and
%\and
%\and
%\and

\author{\IEEEauthorblockN{Yandong Shi,~\IEEEmembership{Student Member,~IEEE}, Hayoung Choi,~\IEEEmembership{Member,~IEEE},\\ Yuanming Shi,~\IEEEmembership{Senior Member,~IEEE}, and Yong
Zhou,~\IEEEmembership{Member,~IEEE}}
        %               ~and Xuemin Shen,~\IEEEmembership{Fellow,~IEEE}}\\
        \thanks{This paper was presented in part at IEEE International Conference on Communications Workshops (ICC Workshops), Dublin, Ireland, 2020 \cite{yandong2020masssive}.
        Y. Shi is with the School of Information Science and Technology, ShanghaiTech University, Shanghai 201210, China, also with the Shanghai Institute of Microsystem and Information Technology, Chinese Academy of Sciences, Shanghai 200050, China, and also with the University of Chinese Academy of Sciences, Beijing 100049, China (e-mail: shiyd@shanghaitech.edu.cn).
        H. Choi is with the Department of Mathematics
        Kyungpook National University Daegu, Republic of Korea (e-mail:hayoung.choi@knu.ac.kr).
   Y. Shi  and Y. Zhou are with the School of Information Science and Technology, ShanghaiTech University, Shanghai 201210, China (e-mail:\{shiym, zhouyong\}@shanghaitech.edu.cn). Yuanming Shi is also with Yoke Intelligence, Shanghai, China.}
        %\thanks{H. Choi is with the Research Institute of Mathematics, Seoul National University, Seoul 08826, Republic of Korea (e-mail:hchoi2@snu.ac.kr).}
        %\thanks{Y. Zhou and Y. Shi are with the School of Information Science and Technology, ShanghaiTech University, Shanghai 201210, China (e-mail:\{zhouyong, shiym\}@shanghaitech.edu.cn).}
        %\thanks{K. B. Letaief is with the Department of Electronic and Computer Engineering,Hong Kong University of Science and Technology, Hong Kong (e-mail:eekhaled@ust.hk). He is also with Peng Cheng Laboratory, Shenzhen, China.}
        %       \thanks{X. Shen is with the Department of Electrical and Computer Engineering, University of Waterloo, Waterloo, ON N2L 3G1, Canada (e-mail: sshen@uwaterloo.ca).}
        %\thanks{This paper has been presented in part at the \textit{IEEE Globecom Workshops}, Waikoloa, Hawaii, Dec. 2019 \cite{Fu2019Intelligent}. }
}

\maketitle

\vspace{-2cm}

\begin{abstract}
Massive access is a critical design challenge of Internet of Things (IoT) networks. 
In this paper, we consider the grant-free uplink transmission of an IoT network with a multiple-antenna base station (BS) and a large number of single-antenna IoT devices. 
Taking into account the sporadic nature of IoT devices,
we formulate the joint activity detection and channel estimation (JADCE) problem as a group-sparse matrix estimation problem.
This problem can be solved by applying the existing compressed sensing techniques, which however either suffer from high computational complexities or lack of algorithm robustness.
To this end, we propose a novel algorithm unrolling framework based on the deep neural network to simultaneously achieve low computational complexity and high robustness  for solving the JADCE problem.
Specifically, we map the original iterative shrinkage thresholding algorithm (ISTA) into an unrolled recurrent neural network (RNN), thereby improving the convergence rate and computational efficiency through end-to-end training.
Moreover, the proposed algorithm unrolling approach inherits the structure and domain knowledge of the ISTA, thereby maintaining the algorithm robustness, which can handle non-Gaussian preamble sequence matrix in massive access.
With rigorous theoretical analysis, we further simplify the unrolled network structure by reducing the redundant training parameters.
%We first introduce an unrolled network structure, named LISTA-GS, by parameterizing the iterative shrinkage thresholding algorithm (ISTA) as an unrolled recurrent neural network (RNN).
%By revealing the weight coupling structure of the former neural network, we further propose a simplified network structure, named LISTA-GSCP, with less training parameters. 
%An analytical network structure named ALISTA-GS is further proposed inspired by theoretical analysis which determines most parameters by solving a simple optimization problem.
Furthermore, we prove that the simplified unrolled deep neural network structures enjoy a linear convergence rate. 
Extensive simulations based on various preamble signatures show that the proposed unrolled networks outperform the existing methods in terms of the convergence rate, robustness and estimation accuracy. 
\end{abstract}


\section{Introduction}\label{sec:intro}

%Massive connectivity will play a crucial role for supporting Internet of Things (IoT) and machine-type communications (MTC) in the future smart networks.

Massive machine-type communications (mMTC), as one of the three major use cases of the fifth-generation (5G) wireless networks \cite{sharma2019towards}, is expected to support ubiquitous connectivity for a large number of low-cost Internet of Things (IoT) devices. 
The typical applications of the mMTC use case include smart homes, wearables, environmental sensing, and healthcare \cite{peng2019smartcity}. 
Driven by the increasing popularity of IoT services and the decreasing cost of IoT devices, the number of IoT devices is expected to reach $ 75.4$ billion by 2025 \cite{chen2020massive}. 
Providing efficient medium access for such a large number of IoT devices is a critical design challenge that needs to be addressed to support mMTC. 

%smart health care, smart home and smart manufacturing \cite{liu2018sparse}.

%To meet the urgent demand of massive connectivity, the new types of multiple access including grant-free random access scheme \cite{liu2018sparse} and non-orthogonal multiple access (NOMA) \cite{dai2015non} are required in mMTC, which can reduce long scheduling delays and large signaling overheads.
%However, in a typical IoT network, a key characteristic of data traffic is the sporadic pattern, which means that only a small and random fraction of devices are active at any time instant \cite{bockelmann2016massive}.
%Thus, along with huge number of devices to be connected, it is of vital importance to support massive device connectivity for IoT networks \cite{chen2020massive}.
%Meanwhile, the base station (BS) has urgent need of the channel state information (CSI) for decoding the uplink signals and precoding the downlink signals before data transmission.
%Thus, it is of vital importance to design efficient \emph{joint activity detection and channel estimation} (JADCE) method for mMTC in IoT networks \cite{chen2020massive}.

%To meet the urgent demand of massive connectivity, the new types of multiple access including grant-free random access scheme \cite{liu2018sparse} and non-orthogonal multiple access (NOMA) \cite{dai2015non} are required in mMTC, which can reduce long scheduling delays and large signaling overheads. 

Grant-free random access, proposed by the 3rd generation partnership project (3GPP) for 5G new radio (NR), has been recognized as a promising multiple access scheme that is capable of accommodating massive IoT devices with low-latency requirements \cite{liu2018sparse}. 
In particular, with grant-free random access, each IoT device can transmit a frame containing both preamble sequence and data at any time instant without waiting for the grant from the base station (BS). 
Compared to the grant-based random access, the signaling overhead and the random access latency can be significantly reduced. 
However, due to the lack of prior knowledge on the subset of active IoT devices that will transmit and the necessity of obtaining accurate channel state information (CSI) for data decoding, it is of vital importance to achieve joint activity detection and channel estimation (JADCE) for grant-free random access in IoT networks\cite{chen2018sparse}. 
%In mMTC scenarios, a key characteristic of data traffic is the sporadic pattern, which means that only a small and random fraction of devices are active at any time instant \cite{bockelmann2016massive}.
%Meanwhile, the base station (BS) has urgent need of the channel state information (CSI) for decoding the uplink signals and precoding the downlink signals before data transmission.
%As a result, \emph{joint activity detection and channel estimation} (JADCE) problem is of great interest for researchers \cite{chen2018sparse}.
%The research on grant-free random access has recently attracted considerable attention. 
To alleviate the collision probability of preamble sequences and in turn facilitate JADCE, the authors in \cite{jiang2019multiple, ding2019comparison, shao2019unified} focused on the design of non-orthogonal preamble signature sequences. 
In particular, the authors in \cite{jiang2019multiple} proposed a novel preamble design that concatenates multiple Zadoff-Chu (ZC) sequences, which increase the success rate of preamble transmission and enhance the channel reuse efficiency for grant-free random access.

Taking into account the sporadic nature of IoT devices, the JADCE problem was formulated as a sparse signal recovery problem in \cite{senel2018grant}. 
This kind of problem can be addressed by adopting the compressed sensing (CS) technique that is capable of exploiting the channel sparsity structures in both the time- and frequency-domains \cite{qin2018sparse, gao2018compressive}. 
%A Bayesian CS-based approach \cite{xu2015active} was proposed to enhance the activity detection performance by exploiting the available information of the large-scale fading. 
Due to the large-scale nature of mMTC, the computationally efficient approximate message passing (AMP) based algorithms were proposed to solve the CS problems for supporting massive access \cite{liu2018massive, chen2018sparse, ke2020compressive}.
The authors in \cite{chen2018sparse} showed that AMP with a vector denoiser and multiple measurement vectors (MMV) \cite{ziniel2012efficient} achieve
approximately the same performance for JADCE.
An AMP algorithm with a minimum mean-squared error (MMSE) denoiser was proposed in \cite{liu2018massive}, where a prior knowledge of the underlying channel distribution was required. 
Furthermore, a generalized multiple measurement vector approximate message passing (GMMV-AMP) algorithm was proposed to detect the device activity by exploiting sparsity in both the spatial and the angular domains \cite{ke2020compressive}.
However, AMP-based algorithms often fail to converge when the preamble signature matrix is mildly ill-conditioned or non-Gaussian \cite{fletcher2018plug}.
To this end, the authors in \cite{jiang2018joint, he2018compressive} introduced a mixed $\ell_1 / \ell_2$-norm to formulate the JADCE problem as a form of group LASSO, which can be solved by using the interior-point method \cite{jiang2018joint} or the alternating direction method of multipliers (ADMM) \cite{he2018compressive} without any prior knowledge on CSI.
In particular, the popular iterative shrinkage thresholding
algorithm (ISTA) \cite{qin2013efficient, bonnefoy2015dynamic} was one of the best known robust algorithmic approaches to solve the group LASSO problem.
However, the ISTA converges only sublinearly in general cases \cite{qin2013efficient} and suffers from an inherent trade-off between the estimation performance and convergence rate \cite{gir2018tradeoff}, which prohibit its practical implementation for grant-free massive access in IoT networks.
Deep learning has recently been emerging as a disruptive
technology to reduce the computational cost.
For instance, the authors in \cite{sun2018learning, ei2019learning, shen2019gnn, shen2020lorm, liang2020learning} proposed the idea of ``learning to optimize", where neural networks were developed to replace the traditional optimization algorithms to solve the non-convex constrained optimization problems \cite{sun2018learning, ei2019learning, shen2019gnn} and mixed-integer nonlinear programming (MINLP) problems \cite{shen2020lorm, liang2020learning} for resource allocation in wireless networks. 
These studies demonstrated the potentials of deep learning for improving communication and computation efficiency.
However, the deep learning framework usually regarded the neural network as a black box and cannot guarantee the theoretical performance. 
Therefore, the lack of interpretability and theoretical guarantee of the deep learning framework is a critical issue that needs to be addressed in wireless networks. 

Fortunately, the unrolled deep neural network \cite{gregor2010learning}, as another powerful deep learning technique, has been proposed to provide a concrete connection between the existing iterative algorithms and the deep neural networks with theoretical performance guarantee \cite{monga2019algorithm, shlezinger2020model}.
The promising applications of this method include resource management, MIMO detection and precoding design \cite{bala2019unfolding, Hu2020unrolling}.
The typical unrolled deep neural network, i.e., \textit{learned iterative shrinkage thresholding
algorithm} (LISTA) \cite{gregor2010learning, liu2018alista, ablin2019learning}, was adopted to approximate the sparse coding based on LASSO, which achieves a better estimation performance.
Moreover, the authors in \cite{chen2018theoretical, liu2018alista} established the simplified structures for LISTA and proved that simplified LISTA achieves linear convergence rate.
The authors in \cite{ablin2019learning} study the selection of adapted step size to improve the convergence rate of ISTA and further propose a unrolled network where only the step sizes of ISTA are learned.
However, such an unrolling framework cannot be directly applied to solve the JADCE problem in IoT networks with grant-free random access, where the uplink channels of active devices are usually represented by a group-row-sparse matrix. 
This is because the existing LISTA with a scalar shrinkage-thresholding operator (SSTO), which only focuses on the recovery of sparse vectors, cannot tackle the group-sparse matrix estimation problem.

\subsection{Contributions}
In this paper, we consider grant-free massive access in
an IoT network, where a large number of single-antenna
IoT devices sporadically transmit non-orthogonal preamble
sequences to a multi-antenna BS. In such a network setting,
we formulate the JADCE problem as a group-sparse matrix
estimation problem by introducing a mixed $\ell_1 / \ell_2$-norm.
Based on the above discussions, we shall develop a new unrolled deep neural network framework by adopting the multidimensional generalization operator, named multidimensional shrinkage thresholding operator (MSTO) \cite{puig2011multidimensional}, to address this problem.
Such an extension turns out to be non-trivial, as the MSTO for the group-row-sparse matrix breaks the independency between the sparse vectors, which brings formidable challenge to reveal the weight coupling structure and analyze the convergence rate for the proposed method.
We address this challenge and provide the theoretical performance guarantee for the proposed framework that is capable of solving JADCE problem to support massive access in IoT networks.

The major contributions of this paper are summarized as follows:
\begin{itemize}
        \item We develop a novel unrolled deep neural network framework with three structures for group-sparse matrix estimation to support grant-free massive access in IoT networks. 
        %Bridging the gap between the research areas of machine learning and signal processing for massive access, 
        The proposed methods enjoy high robustness by inheriting the structure of ISTA and keep the computational complexity at a low level through end-to-end training.
        %We propose three unrolled neural network structures to tackle the group-sparse matrix estimation problem. We conduct rigid theoretical analysis and reveal the weight coupling structure and its property which depends on the preamble signature. 
        \item We conduct rigorous theoretical analysis to get rid of the interpretability issue and facilitate the efficient training in IoT networks. 
        We reveal the weight coupling structure between the training parameters and identify the property that the learned weight matrix only depends on the preamble signature.
        %the weight parameters' property which only depends on the preamble signature.
        The theoretical analysis allows us to further simplify original unrolled network structure by reducing the redundant training parameters.
Furthermore, we prove that the simplified unrolling neural networks enjoy a linear convergence rate. 
% based on the method of classification discussion.
        %To tackle the interpretability issue, we prove the convergence rate of the proposed unrolled networks for group estimation problem is linear.
        %The extension turns out to be nontrivial since MSTO of group row sparse matrix breaks up the \textquotedblleft independency" between sparse vectors.
%       Our simulation experiments support these theoretical results.
        \item Extensive simulations are conducted to validate the theoretical analysis and show the superior performance of the proposed algorithms in three critical aspects for grant-free massive access. 
        Firstly, comparing with the robust algorithms such as ISTA, the proposed methods have much lower computational cost.
        Secondly, the proposed methods are more robust when comparing to computational-efficient algorithms such as AMP-based algorithms by considering simulating different preamble signatures, including Gaussian matrix, binary matrix, and ZC sequence.
        Thirdly, the proposed algorithms are capable of returning more accurate solutions comparing with the classical CS-based algorithms. 
        %Moreover, we demonstrate the robustness of the proposed algorithms by simulating three different preamble signatures, including Gaussian matrix, binary matrix, and ZC sequence. 
\end{itemize}

\subsection{Organization and Notations}
The rest of this paper is organized as follows.
The system model and problem formulation are described in Section \ref{sec:sys}. 
Section \ref{sec:alg} presents the proposed three kinds of unrolled neural networks for solving the group-sparse matrix estimation problem.
The convergence analysis of the proposed unrolled neuron networks is presented in Section \ref{sec:con}.
The numerical results of the proposed algorithms are illustrated in Section \ref{sec:num}.
Finally, Section \ref{sec:conclu} concludes this paper. 

\subsubsection*{Notations}
Throughout this paper, we denote $[N]=\{1,2,\ldots,N \}$. 
Let $\R$ (resp. $\C$) be the set of real (resp. complex) numbers and $\mathbb{N}$ be the set of integers.
For a set $\mathcal{S}$, the number of elements of $\mathcal{S}$ is denoted by $|\mathcal{S} |$.
For index sets $\mathcal{I} = \{ i_1, \ldots, i_n \} \subset [N]$, $\mathcal{J} = \{ j_1,\ldots, j_m \} \subset [M]$ and $\bm{A} \in \R^{N \times M}$, we denote by  $\bm{A}[\mathcal{I}, \mathcal{J} ]$ the (sub)matrix of entries that lie in the rows of $\bm{A}$ indexed by $\mathcal{I}$ and the columns indexed by $\mathcal{J}$.
When $\mathcal{I} = [N]$ (resp. $\mathcal{J} = [M]$), we simply denote by $\bm{A}[:, \mathcal{J} ]$ (resp. $\bm{A}[\mathcal{I}, : ]$).
%When $\mathcal{I} = [N]$ (resp. $\mathcal{J} = [M]$) , we simply denote by $\bm{A}[:, \mathcal{J} ]$ (resp. $\bm{A}[\mathcal{I}, : ]$).
For a matrix $\bm{A}$,
$\|\bm{A}\|_2$ (resp. $\|\bm A\|_F$) denotes the operator (resp. Frobenius) norm.
And $\|\bm{A}\|_{\max}$ denotes the max norm.
Also,  $\|\bm{A} \|_{2,1}$ (resp. $\|\bm{A} \|_{2,0}$) denotes mixed $ \ell_1/\ell_2$-norm (resp. mixed $ \ell_0/\ell_2$-norm).
For a vector $\bm{v}$, $\| \bm{v} \|_2$ denotes the Euclidean norm.
%The superscript $T$ denotes the transpose operator.
%$\bm{A}^{H} = \bar{\bm{A}}^{T}$ is a complex conjugate transpose of $\bm{A}$.
We denote the support of the vector $ \bm{v}=[v_1,\ldots,v_N] \in \R^N$ as $\text{supp}(\bm{v})$, i.e.  
$\text{supp}(\bm{v}) = \{ i\in [N] | v_i \neq 0 \}$.

%\bigskip

%$\mathcal{S} = \text{supp}(\bm{v})$, and let $ | \mathcal{S} | $ denote the number of elements in the set $\mathcal{S}$.

\section{System Model and Problem Formulation}\label{sec:sys}

\subsection{System Model}
Consider the grant-free uplink transmission of a single-cell IoT network consisting of one $M$-antenna BS and $N$ single-antenna IoT devices, as shown in Fig. \ref{fig:model}. 
Full frequency reuse and quasi-static block fading channels are considered in this paper. 
We denote $[N] = \{1, \ldots, N \}$ as the set of IoT devices, which sporadically communicate with the BS. 
In each transmission block, all the IoT devices independently decide whether or not to be active. 
Within a specific transmission block, we denote $a_n = 1$ if device $n$ is active and $a_n = 0$ otherwise. 
Because of the sporadic transmission, it is reasonable to assume that each IoT device is active with a probability being much smaller than 1. 
In addition, we assume that the active IoT devices are synchronized at the symbol level \cite{jiang2018joint}.
The preamble signal received at the BS, denoted as $\bm{y}(\ell) \in \mathbb{C}^M$, can be expressed as 
\begin{eqnarray}\label{rs}
\bm y(\ell) = \sum_{n=1}^{N} \bm{h}_n a_n s_n(\ell) + \bm{z}(\ell) ,\quad \ell = 1,...,L, 
\end{eqnarray}
where $\bm{h}_n \in {\C}^{M}$ denotes the channel vector between device $n$ and the BS, $s_n(\ell) \in \mathbb{C}$ denotes the $\ell$-th symbol of the preamble signature sequence transmitted by device $n$, $L$ denotes the length of the signature sequence, and $\bm{z}(\ell) \in \mathbb{C}^M$ denotes the additive white Gaussian noise (AWGN) vector at the BS. 

\begin{figure}[t]
        \centering
        \begin{minipage}{.46\textwidth}
                \centering
                \includegraphics[width=1.0\columnwidth]{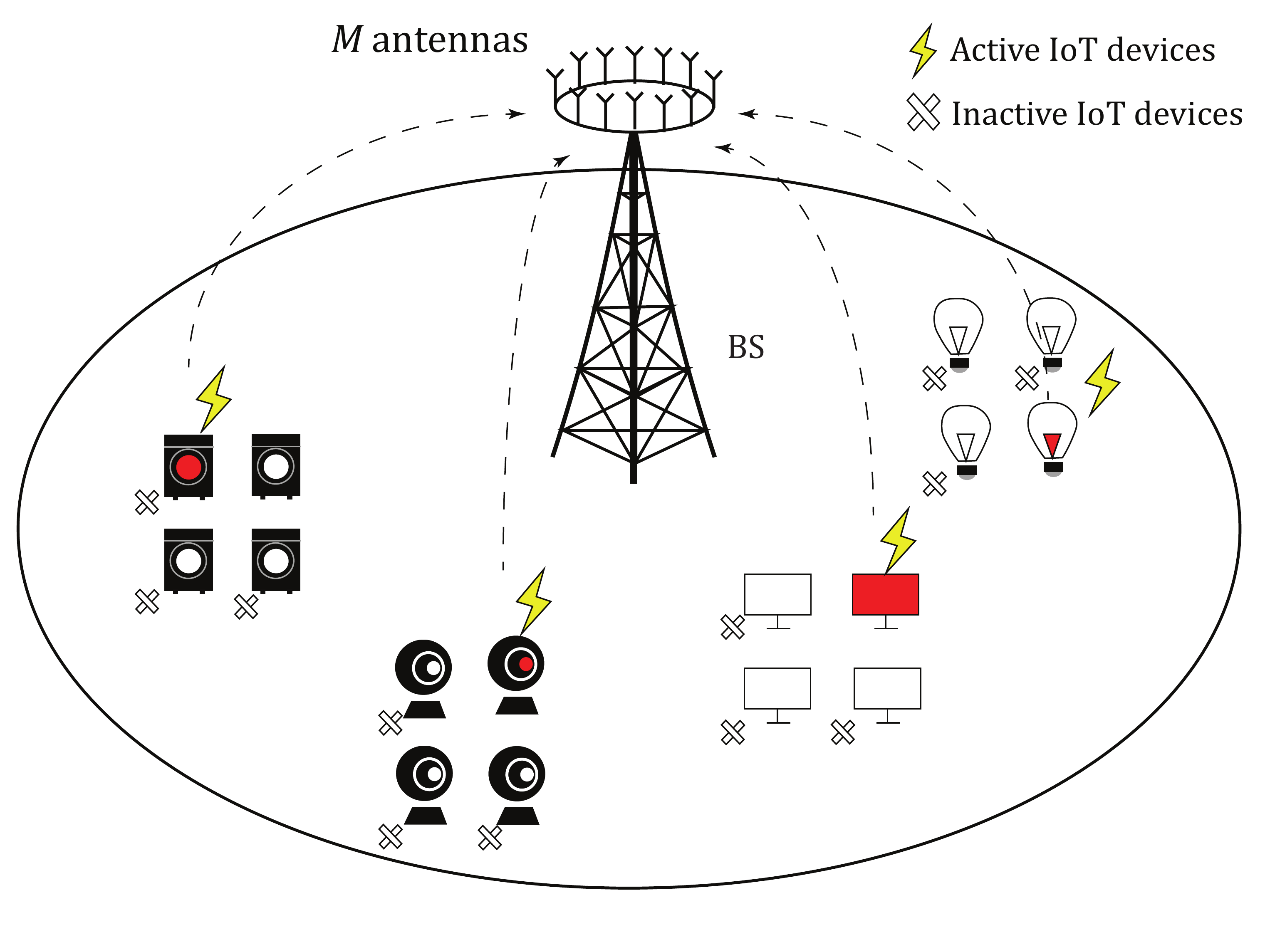}
                \caption{Illustration of an IoT network with a massive number of IoT devices that transmit sporadically to the BS.}\label{fig:model}
        \end{minipage}
        \hspace{4mm}
        \begin{minipage}{.46\textwidth}
                \centering
                \includegraphics[width=1.0\columnwidth]{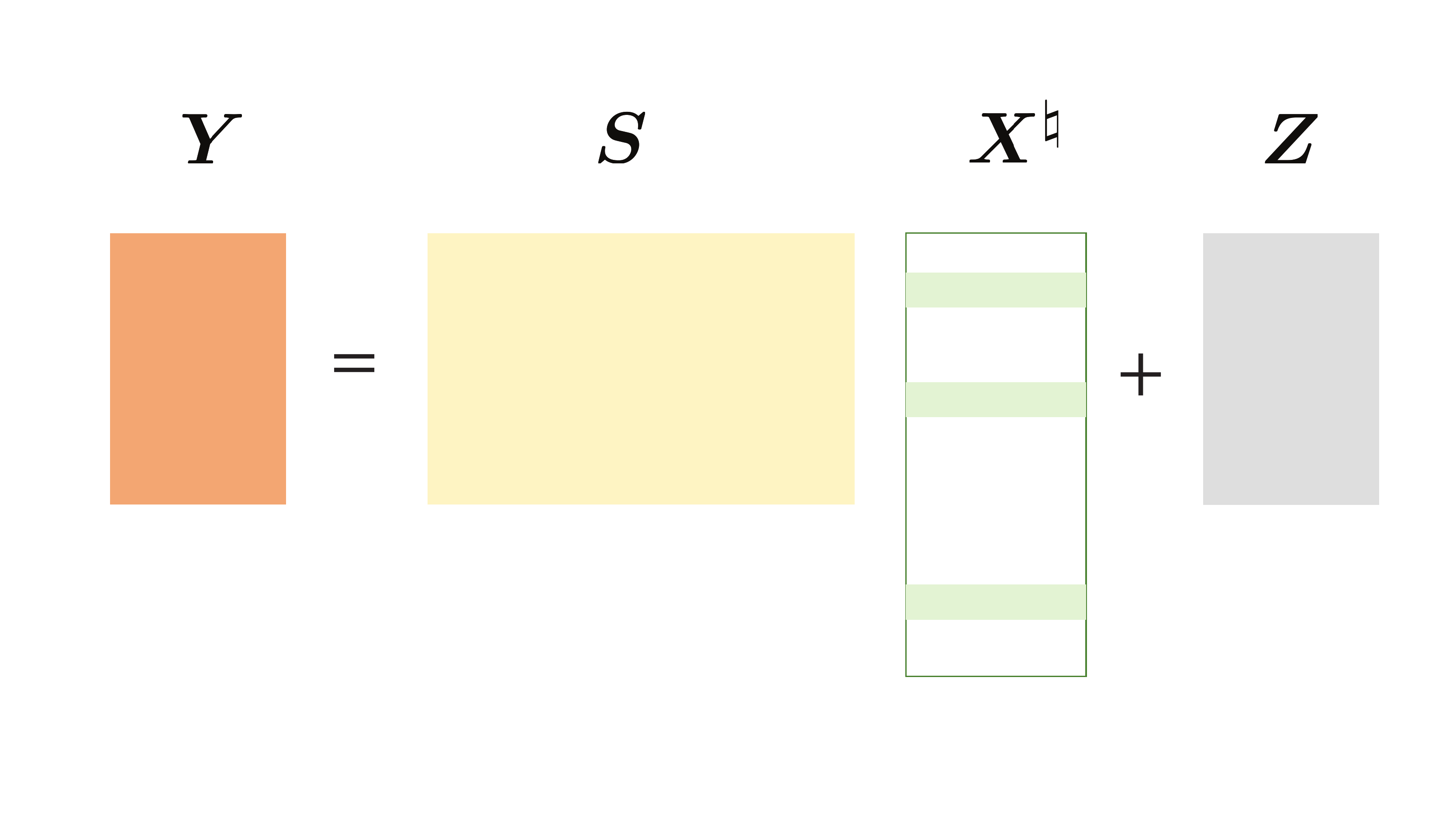}
                \caption{Illustration of the JADCE problem, where $\bm{X}^{\natural}$ has the group-sparse structure in rows, i.e., if one entry of $\bm{X}^{\natural}$ is zero, then other entries on the same row should  be zero simultaneously. }\label{fig:gp}
        \end{minipage}
\end{figure}
As the length of the signature sequence is typically much smaller than the number of IoT devices, i.e., $L \ll N$, it is impossible to allocate mutually orthogonal signature sequences to all the IoT devices. 
Hence, each IoT device is assigned a unique signature sequence, which is generally non-orthogonal to the preamble sequences assigned to other IoT devices.  
Note that three different kinds of non-orthogonal preamble sequences will be considered in the simulations.

For notational ease, we denote $\bm{Y} = [\bm{y}(1),...,\bm{y}(L)]^T \in \C^{L \times M}$ as the received preamble signal matrix across $M$ antennas over the transmission duration of $L$ symbols, $\bm{H} = [\bm{h}_1,...,\bm{h}_N]^T \in \C^{N \times M} $ as the channel matrix between $N$ devices and the BS, $\bm{Z} = [\bm{z}(1),...,\bm{z}(L)] \in \C^{L \times M}$ as the noise matrix at the BS, and $\bm S= [\bm{s}(1),...,\bm{s}(L)]^T \in \C^{L \times N} $ as the known preamble signature sequence matrix with $\bm s(\ell) = [s_1(\ell),...,s_N(\ell)]^T \in \C^N$. As a result, the preamble signal received at the BS can be rewritten as 
\begin{eqnarray}\label{ms}
\bm{Y} = \bm{SAH} + \bm{Z},
\end{eqnarray}
where $\bm{A} = {\diagg}(a_1,...,a_N) \in \R^{N \times N}$ denotes the diagonal activity matrix. 
We aim to simultaneously detect the activity matrix $\bm A$ and estimate the channel matrix $\bm H$, which is known as a JADCE problem \cite{liu2018massive}.
By denoting $ \bm{X}^{\natural} = \bm{AH} \in \C^{N \times M}$, matrix $\bm{X}^{\natural}$ thus has the group-sparse structure in rows \cite{jiang2018joint}, as shown in Fig. \ref{fig:gp}.
We further rewrite (\ref{ms}) as 
\begin{eqnarray}\label{mmv}
\bm{Y} = \bm{SX}^{\natural} + \bm{Z}.
\end{eqnarray} 
Note that the active devices and their corresponding channel vectors can be recovered from the estimation of matrix $\bm{X}^{\natural}$ \cite{jiang2018joint}.

\subsection{Problem Formulation}
To estimate the group-row sparse matrix $\bm{X}^{\natural}$, we adopt the mixed $ \ell_1/\ell_2$-norm \cite{jiang2018joint} to induce the group sparsity, i.e., 
$
\mathcal{R}(\bm{X}) = \sum_{n=1}^{N} \Vert \bm{X}[n,:] \Vert_2, 
$
where $\bm{X}[n,:]$ is the $n$-th row of matrix $\bm{X}$.
The group-sparse matrix estimation
problem can be reformulated as the following unconstrained convex optimization problem (also known as group LASSO) \cite{yuan2006model}
\begin{eqnarray}\label{probi}
\mathscr{P}:  \mathop{\mini}\limits_{\bm{X} \in \C^{N\times M}} \frac{1}{2} \Vert \bm{Y} - \bm{S}\bm{X} \Vert_{F}^2 + \lambda\mathcal{R}(\bm{X}), 
\end{eqnarray}
where $\lambda > 0$ denotes the regulation parameter.
As matrix $\bm{X}$ is group sparse in rows, problem $\mathscr{P}$ is essentially
a group-sparse matrix estimation problem.
%
%Note that the device activity matrix $\bm A$ and the channel matrix $\bm H$ can be recovered from the estimate of matrix $\bm X$. %based on observation matrix $\bm{Y}$ as $ \bm{X} = \bm{AH}$. 

To facilitate efficient algorithm design, we rewrite (\ref{mmv}) as its real-valued counterpart
\begin{eqnarray}\label{realmmv}
        \begin{aligned}
        \bm{\tilde{Y}}  &= \bm{\tilde{S}} \bm{\tilde{X}}^{\natural} + \bm{\tilde{Z}} \\
        &= \left[ \begin{matrix}
        \Re\{\bm{S}\} & -\Im\{\bm{S}\} \\
        \Im\{\bm{S}\} & \Re\{\bm{S}\}
        \end{matrix} \right] \left[ \begin{matrix}
        \Re\{ \bm{X}^{\natural} \} \\
        \Im \{\bm{X}^{\natural}\}
        \end{matrix} \right] +   \left[ \begin{matrix}
        \Re \{ \bm{Z} \} \\
        \Im \{ \bm{Z} \}
        \end{matrix} \right],
        \end{aligned}
\end{eqnarray}
where $\Re\{ \cdot \}$ and $\Im\{ \cdot \}$ represent the real part and imaginary part of a complex matrix.
As a result, problem $\mathscr{P}$ can be further rewritten as
\begin{eqnarray}\label{groupLASSO}
        \mathscr{P}_r: \mathop{\mini}\limits_{\bm{\tilde{X}} \in \R^{2N\times M}} \frac{1}{2} \Vert \bm{\tilde{Y}} - \bm{\tilde{S}}\bm{\tilde{X}} \Vert_{F}^2 + \lambda\mathcal{R}(\bm{\tilde{X}}).
\end{eqnarray}

Therefore, our goal of JADCE becomes the recovering of matrix $\bm{\tilde{X}}^{\natural} $ based on the preamble matrix $ \bm{\tilde{S}}$ and the noisy observation $ \bm{\tilde{Y}}$.
\newtheorem{assumption}{Assumption}\label{ass1}
%\begin{assumption}
%        The group-row-sparse signal $ \bm{\tilde{X}}^{\natural}$ and the observation noise $\bm{\tilde{Z}}$ are sampled from a set $\mathcal{X} (\beta, s, \sigma) $ fulfilling:
%        \begin{eqnarray}
%        \mathcal{X} (\beta, s, \sigma) 
%        \coloneqq  
%        \big\{ (\bm{\tilde  X}^{\natural}, \bm{\tilde{Z}} )|  \| \bm{\tilde{X}}^{\natural}[i,:] \|_2 \leq \beta, &\forall i,
%        \| \bm{\tilde  X}^{\natural}\|_{2,0} \leq s, \notag \|\bm{\tilde{Z}} \|_{F} \leq \sigma \big\}. \label{eq:assumption1}
%        \end{eqnarray}
%\end{assumption}

\subsection{Prior Work}
The ISTA \cite{yuan2006model} is a popular approach to solve the group LASSO problem.
In particular, the resulting ISTA iteration for \eqref{groupLASSO} can be written as
\begin{eqnarray}\label{ISTA}
\bm{\tilde{X}}^{k+1} = \eta_{\lambda / C} \left(\bm{\tilde{X}}^k + \frac{1}{C}\bm{\tilde{S}}^T(\bm{\tilde{Y}} - \bm{\tilde{S}}\bm{\tilde{X}}^k) \right),
\end{eqnarray}
where $\bm{\tilde{X}}^k $ is an estimate of ground truth $\bm{\tilde{X}}^{\natural}$ at iteration $k$, $\frac{1}{C}$ plays a role as step size, $\lambda$ is regulation parameter, and $\eta_{\lambda / C}(\cdot)$ denotes the MSTO \cite{puig2011multidimensional}.
Specifically, $\eta_{\theta}(\bm{X})[n] $ denotes as the $n$-th row of the matrix $\bm{X}$ after applying $\eta_{\theta}(\cdot)$, which is defined as \cite{bonnefoy2015dynamic}
\begin{eqnarray}
\eta_{\theta}(\bm{X})[n] = \max \left\{0, \frac{\Vert \bm{X}[n,:] \Vert_2 - \theta}{\Vert \bm{X}[n,:] \Vert_2 } \right\} \bm{X}[n,:].
\end{eqnarray}
%where $\eta_{\theta}(\bm{X})[n] $ denotes the $n$-th row in the matrix after applying $\eta_{\theta}(\cdot)$.

%However, the ISTA for the group-sparse matrix estimation problem (ISTA-GS) converges only sublinearly in general cases which is somewhat slow \cite{qin2013efficient}.
However, ISTA suffers from an inherent trade-off between estimation performance and convergence rate based on the choice of $\lambda $, i.e., a larger $\lambda $ leads to faster convergence but a poorer estimation performance \cite{gir2018tradeoff}.
Moreover, the choice of step size $ \frac{1}{C}$ also influences the convergence rate \cite{ablin2019learning}.
In the next section, we shall propose a learned ISTA framework by learning the parameters including $\lambda $ and the step size $ \frac{1}{C}$ simultaneously to improve the convergence rate and estimation performance of ISTA for the group-sparse matrix estimation problem (ISTA-GS).

\section{Algorithm Unrolling via Deep Neural Network}\label{sec:alg}
In this section, we propose an unrolled deep learning framework to solve the JADCE problem by exploiting the group sparse structure.
Although the LISTA proposed in \cite{gregor2010learning, chen2018theoretical, liu2018alista, ito2019trainable} can recover the individual sparse vector signals, all these methods cannot recover the matrix with group row sparsity.
To address this issue, we extend LISTA for the group-sparse matrix recovery.
Furthermore, we analytically prove that the weight coupling property holds for the group-sparse matrix estimation problem.
\begin{figure}
        \centering
        \subfigure[Unrolled network strucrure 1 (LISTA-GS).]{
                \includegraphics[width=1\linewidth]{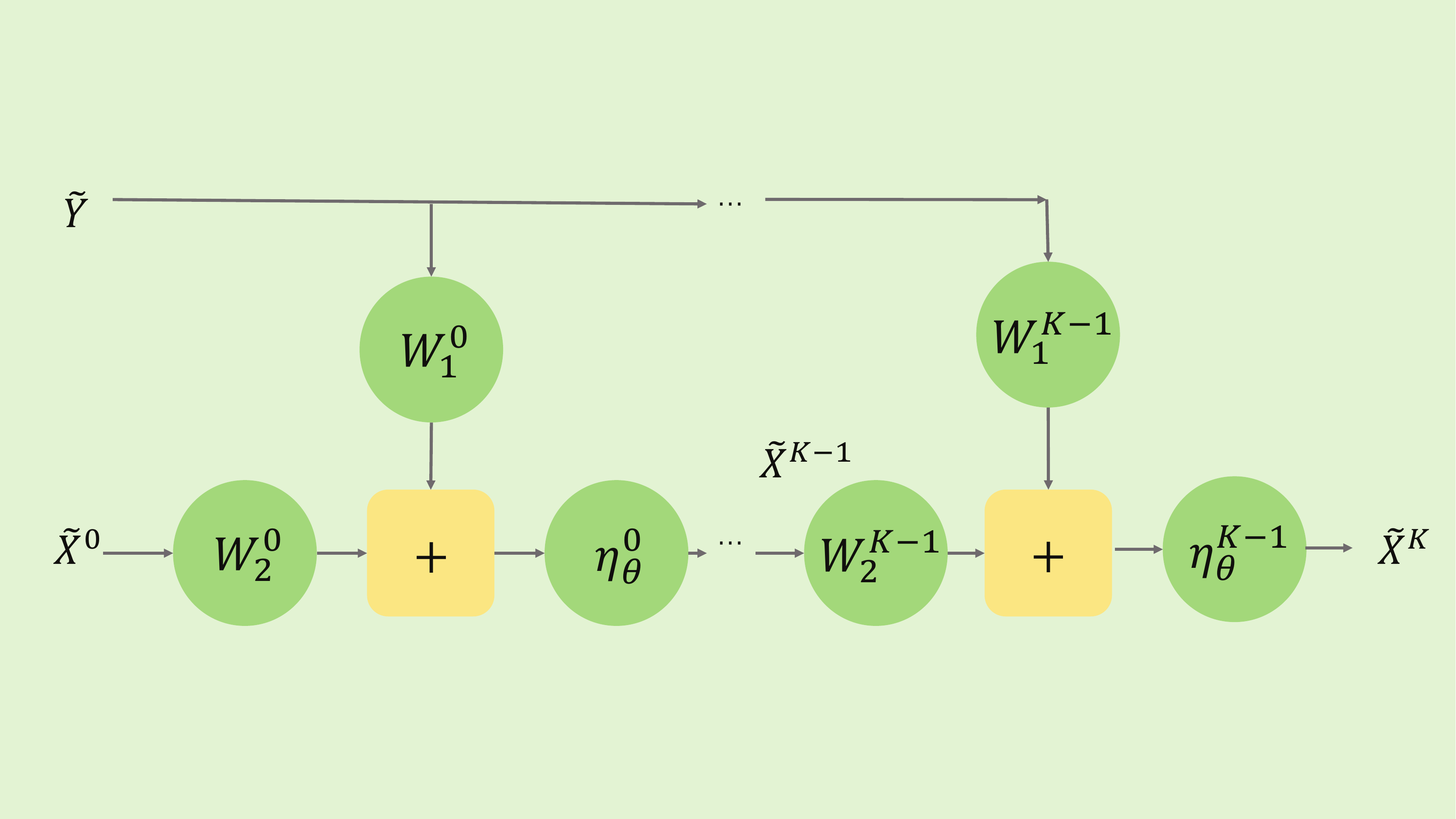}
                        \label{fig:a}
        }
        \subfigure[Unrolled network structure 2 (LISTA-GSCP).]{
                \includegraphics[width=1\linewidth]{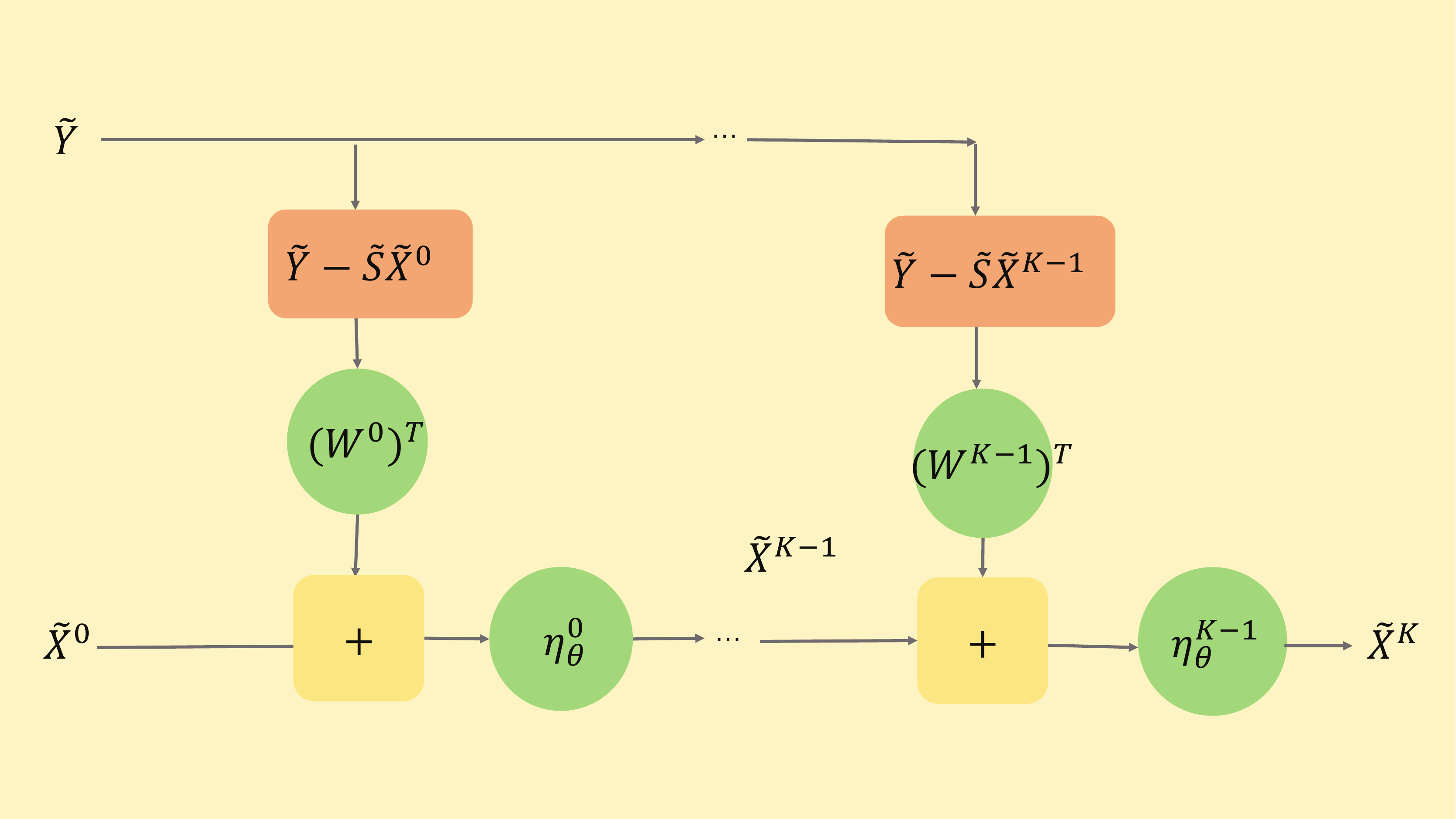}
                        \label{fig:b}
        }
        \subfigure[Unrolled network structure 3 (ALISTA-GS).]{
                        \includegraphics[width=1\linewidth]{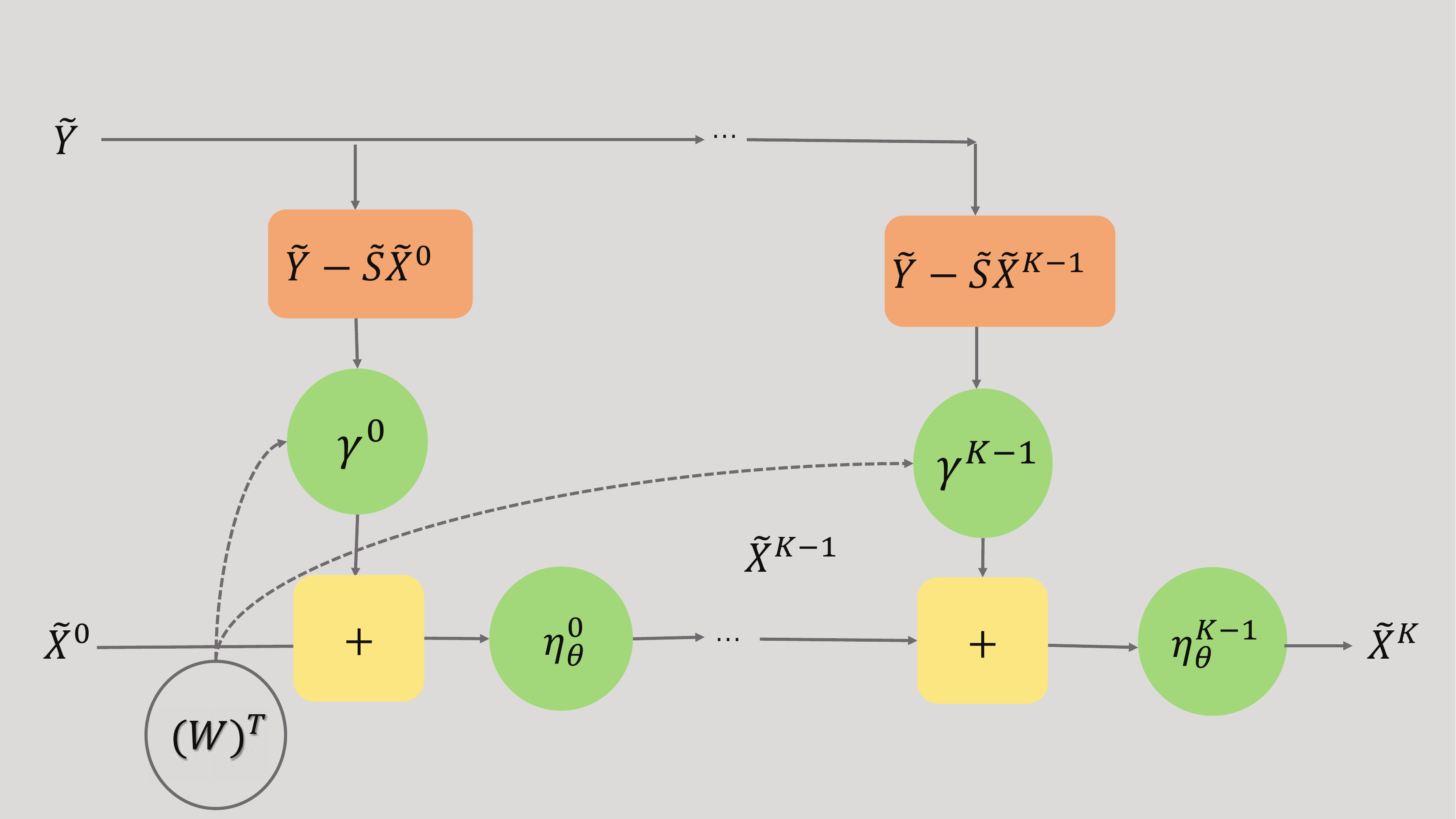}
                        \label{fig:c}
        }
        \caption{Three unrolled network structures. }
\end{figure}

\subsection{Unrolled Neural Network Structures}
In this subsection, we propose three neural network structures under the unrolled framework for group-sparse matrix estimation problem $\mathscr{P}_r$.
%\begin{itemize}
        %\item \textbf{Unrolled network 1 (namely LISTA-GS)}: The original framework of LISTA for group sparse estimation. The unfolded neural network structure is expressed in \eqref{eq:LISTA1}.
        %\item \textbf{Unrolled network 2 (namely LISTA-GSCP)}: The partial weight coupling and weight sharing of LISTA for group sparse estimation.
        %\eqref{eq:LISTA-CP} reveals its network structure. 
        %\item \textbf{Unrolled network 3 (namely ALISTA-GS)}: The LISTA framework that only trains the step size and threshold scalar. We will denote it in \eqref{eq:ALISTA}.
%\end{itemize}
\subsubsection{LISTA-GS}
Inspired by \cite{gregor2010learning, chen2018theoretical} and by denoting $ \bm{W}_1 = \frac{1}{C}\bm{\tilde{S}}^T$, $\bm{W}_2 = \bm{I} -  \frac{1}{C}\bm{\tilde{S}}^T \bm{\tilde{S}}$, and $\theta = \frac{\lambda}{C}$,
we rewrite (\ref{ISTA}) as
\begin{eqnarray}\label{eq:ista}
\bm{\tilde{X}}^{k+1} = \eta_{\theta^k}(\bm{W}_1 \bm{\tilde{Y}} + \bm{W}_2 \bm{\tilde{X}}^{k}). 
\end{eqnarray}

The key idea of the proposed unrolled method for group-sparse matrix estimation problem is to view matrices $\bm{W}_1$, $\bm{W}_2$, and scalars $\theta^k$ in (\ref{eq:ista}) as trainable parameters. 
As a result, (\ref{eq:ista}) can be modeled as a one layer recurrent neural network (RNN), i.e., one iteration of ISTA-GS is treated as one layer neural network.
The neural network structure can be concatenated as a $K$-layer RNN, which is capable of modeling the ISTA-GS with $K$ iterations.
%The whole map of unrolled framework is built as the process going layer by layer.
Mathematically, the unrolled RNN with $K$ iterations for group-sparse matrix estimation problem is given by
\begin{eqnarray}\label{eq:LISTA1}
\bm{\tilde{X}}^{k+1} = \eta_{\theta^k}(\bm{W}_1^k \bm{\tilde{Y}} + \bm{W}_2^k \bm{\tilde{X}}^{k}), k = 0, 1, \ldots, K-1, 
\end{eqnarray}
where parameters $\bm {\Theta} = \{ \bm{W}_1^k,\bm{W}_2^k, \theta^k \}_{k=0}^{K-1} $ are trainable. This is the main difference from the iterative algorithm in (\ref{ISTA}). 

The unrolled neural network structure given in \eqref{eq:LISTA1} is named as LISTA-GS, and the corresponding structure is plotted in Fig. \ref{fig:a}. However, this kind of structure contains too many parameters with two high-dimensional weight parameter matrices, which may not be efficient to be trained.

%As similar in \cite{chen2018theoretical}
\subsubsection{LISTA-GSCP}

We establish the following necessary condition for the convergence of LISTA-GS, which is inspired by the fact that $\bm{W}_2^0 = \bm I-\bm{W}_{1}^{0}\bm{\tilde{S}}$ and coupled structure in \cite{chen2018theoretical}.
This necessary condition reveals the properties of trainable parameters if the proposed unrolled NN can recover the group sparse matrix and can be used to simplify the proposed network.
\begin{theorem}\label{thm:main1}
        Given $\{\bm{W}_{1}^{k}, \bm{W}_{2}^{k},   \theta^{k}\}_{k=0}^{\infty}$ with bounded weights, i.e., $   \|\bm{W}_{1}^{k}\|_F \leq C_{W_1}$ and $\| \bm{W}_{2}^{k} \|_F \leq C_{W_2}$, let $\{\bm{\tilde  X}^{k}\}_{k=1}^{\infty}$ be generated layer-wise by LISTA-GS in \eqref{eq:LISTA1} with an input $\bm{\tilde{Y}}$ observed by 
        \eqref{mmv} and the initial point $\bm{\tilde  X}^{0}=\bm{0}$.
         If network can recover any group-row-sparse signal with no observation noise, then $\{\bm{W}_{1}^{k},\bm{W}_{2}^{k},   \theta^{k}\}_{k=0}^{\infty}$ must satisfy the following two conditions
\begin{align}   
& \theta^{k}  \rightarrow 0 , \quad \text {as } k \rightarrow \infty,  \label{thm1:result1}\\
&\bm{W}_{2}^{k}-(\bm I-\bm{W}_{1}^{k} \bm{\tilde{S}}) \rightarrow \bm{0} , \quad \text{as } k \rightarrow \infty. \label{thm1:result2} 
\end{align}
\end{theorem}

\begin{proof}
        Please refer to Appendix \ref{pr:th1}.
\end{proof}

Thus, we develop the necessary condition of the recovery convergence for the group-sparse matrix estimation problem. 
It is worth nothing that the MSTO \cite{puig2011multidimensional} is a generalization of the SSTO, which brings unique challenge of revealing the weight coupling structure on the group-row-sparse matrix estimation problem.
The extension turns out to be non-trivial since the MSTO of group-sparse matrix breaks up the 
\textquotedblleft independency" between sparse vectors.

%The extension turns out to be nontrivial since group soft-thresholding function break the weight coupling structure for unrolled network in individual sparsity.
%Theorem \ref{thm:main1} builds necessary conditions of recovery convergence, and discovers the relationship between the weight matrices $\bm{W}_1,\bm{W}_2$ in unfolded LISTA network structure.
%Theorem 1 endorses the empirical results shown in Fig. \ref{fig:tha}, \ref{fig:thb}.
Motivated by (\ref{thm1:result2}) in Theorem \ref{thm:main1}, we adopt the partial weight coupling structure of the trainable weights $\{ \bm{W}_1^k,\bm{W}_2^k \}_{k=0}^{\infty} $ in LISTA-GS as
$
        \bm{W}_2^k = \bm I-\bm{W}_{1}^{k}\bm{\tilde{S}},  \forall k.
$
Hence, by letting $ \bm{W}^k = (\bm{W}_1^k)^T$, we obtain the simplified $K$-layer unrolled neural network structure (namely LISTA-GSCP) as
\begin{equation}\label{eq:LISTA-CP}
\bm{\tilde{X}}^{k+1} = \eta_{\theta^k} \big( \bm{\tilde{X}}^k + (\bm{W}^k)^T (\bm{\tilde{Y}} - \bm{\tilde{S}} \bm{\tilde{X}}^k) \big), k = 0, 1, \ldots, K-1, 
\end{equation}
where parameters $\bm {\Theta} = \{ \bm{W}^k,  \theta^k  \}_{k=0}^{K-1} $ are trainable. Fig. \ref{fig:b} illustrates the unrolled network.

\subsubsection{ALISTA-GS}
To further alleviate the need to learn a weight matrix $\bm{W}^k $ with a large amount of parameters while stabilizing the training process, we separate the weight as the product of a scalar and a matrix, i.e., $\bm{W}^k = \gamma^k \bm{W} $. 
The definition of  \textquotedblleft good" parameters $\bm{W} \in \mathcal{W}(\bm{\tilde{S}})$ in \eqref{g3} shows that the weight matrix only depends on $ \bm{\tilde{S}}$.
Thus, we can solve the following convex optimization problem by using the projected gradient descend (PGD) algorithm to obtain the weight matrix $\bm{W}$ prior to the training stage
\begin{align}
        &\underset{\bm{W} \in \R^{2L \times 2N}}{\mini} \quad \| \bm{W}^T \bm{\tilde{S}} \|_{F}^2 \notag \\
        &\subj \quad \bm{W}[:,i]^T \bm{\tilde{S}}[:,i] = 1, \forall i \in [2N].
\end{align}
Hence, the proposed third unrolled network structure (namely ALISTA-GS) comes to light
\begin{equation}\label{eq:ALISTA}
\bm{\tilde{X}}^{k+1} = \eta_{\theta^k} \big( \bm{\tilde{X}}^k + \gamma^{k} \bm{W}^T (\bm{\tilde{Y}} - \bm{\tilde{S}} \bm{\tilde{X}}^k) \big), k = 0, 1, \ldots, K-1, 
\end{equation}
where $\bm {\Theta} = \{ \theta^k ,  \gamma^k   \}_{k=0}^{K-1} $ are parameters to be trained and the weight matrix $\bm{W}$ is pre-computed before the training process. 
Different from the LISTA-GSCP, this network structure pays more attention to select better learning step sizes.
A visual depiction is provided in Fig. \ref{fig:c}.

Till now, we have established three LISTA variations for group-sparse matrix estimation problems. With less training parameters, the training process turns to be easier and faster. 
\begin{remark}
Recently, the advancement of complex-valued deep networks make it possible to apply the proposed LISTA variations in a complex-valued representation, i.e., without rewriting as its real-valued counterpart (\ref{realmmv}).
The complex-valued representation may be utilized to further enhance the performance of sparse recovery by exploiting the extra information about potential grouping of real and imaginary parts, e.g., the real and imaginary parts are either zero or nonzero simultaneously \cite{male2013camp, ta2020ctista}.
However, the design of the complex-valued network training process and the complex-valued shrinkage function will be two challenges.
Besides, the convergence analysis of such complex-valued LISTA variations will be an interesting problem, which will be studied in our future work.
\end{remark}

\subsection{Deep Neural Networks Training and Testing}
In the training stage, we consider the framework that the neural networks learn to solve group-sparse matrix estimation problem $\mathscr{P}_r$ in a supervised way.
Note that we denote $P$ samples training data as $ \{ \bm{\tilde{X}}_i^{\natural}, \bm{\tilde{Y}}_i\}_{i=1}^P$, where $ \bm{\tilde{X}}_i^{\natural}$ and $\bm{\tilde{Y}}_i $ are viewed as the label and the input, respectively. 
The following optimization problem is adopted in the whole training stage for training a $k$-layer network,
\begin{equation}\label{train}
\mathop{\mini}\limits_{\bm \Theta_{0:k-1}}~ \sum_{i=1}^{P} \Vert \bm{\tilde{X}}^k(\bm \Theta_{0:k-1}, \bm{\tilde{Y}_i}, \bm{\tilde{X}}^{0}) - \bm{\tilde{X}}_i^{\natural} \Vert_F^2.
\end{equation}
However, such large scale optimization easily converged to a bad local minimum \cite{bor2017amp}.
The layer-wise training strategy is used to separate (\ref{train}) into the following two parts, i.e., (\ref{tr:1}) and (\ref{tr:2}), where (\ref{tr:1}) is to find a good initialization of (\ref{tr:2}) to avoid bad local minima of (\ref{train}).
In particular, for training a $k$-layer network, we denote $ \bm{\Theta}_{0:k-1}$ as the trainable parameters from layer $0$ to layer $k-1$, and $ \bm{\Theta}_{k-1}$ as the trainable parameters in layer $k-1$. 
By fixing the parameters $\bm{\Theta}_{0:k-2} $ of the former $k-1$-layer network that has been trained and the parameters $\bm{\Theta}_{k-1} $'s initialization is given in Table \ref{ta:1}, we first learn the parameters $\bm{\Theta}_{k-1} $ with the learning rate $\alpha_0 $ by using Adam algorithm \cite{kabashima2003cdma} to solve the following optimization problem
\begin{equation}\label{tr:1}
\mathop{\mini}\limits_{\bm \Theta_{k-1}}~ \sum_{i=1}^{P} \Vert \bm{\tilde{X}}^k(\bm \Theta_{0:k-1}, \bm{\tilde{Y}}_i, \bm{\tilde{X}}^{0}) - \bm{\tilde{X}}^{\natural} \Vert_F^2,
\end{equation}
where $\bm{\tilde{X}}^k(\bm \Theta_{0:k-1}, \bm{\tilde{Y}}, \bm{\tilde{X}}^{0})$ denotes the output of the $k$-layer network with input $\bm{\tilde{Y}}$ and initial point $ \bm{\tilde{X}}^{0}$.
Then, we use the parameters $\bm{\Theta}_{k-1}$ obtained by \eqref{tr:1} and the fixed parameters $\bm{\Theta}_{0:k-2} $ as initialization and then tune all parameters $ \bm{\Theta}_{0:k-1}$ with the learning rate $\alpha_1 $ which is smaller than  $\alpha_0 $ by using Adam algorithm to solve
\begin{equation}\label{tr:2}
\mathop{\mini}\limits_{\bm \Theta_{0:k-1}}~ \sum_{i=1}^{P} \Vert \bm{\tilde{X}}^k(\bm \Theta_{0:k-1}, \bm{\tilde{Y}}_i, \bm{\tilde{X}}^{0}) - \bm{\tilde{X}}^{\natural} \Vert_F^2.
\end{equation}
After applying the procedure successfully, we obtain the parameters $ \bm{\Theta}_{0:k-1}$ for the whole $k$-layer network.
So next we can train a $(k+1)$-layer network.

In the testing stage, since the known preamble signature $\bm{S}$ remains unchanged, the proposed unrolled networks can recover the changing channels.
Given a newly received signal $\bm{\tilde{Y}}^{'} $, the learned unrolled networks, i.e., LISTA-GS in \eqref{eq:LISTA1}, LISTA-GSCP in \eqref{eq:LISTA-CP} and ALISTA-GS in \eqref{eq:ALISTA}, are applied for group-sparse matrix estimation.
For instance, with LISTA-GS, we obtain the $(k+1)$-th layer recovered signal by using $\bm{\tilde{X}}^{k+1} = \eta_{(\theta^k)^{*}}((\bm{W}_1^k)^{*} \bm{\tilde{Y}}^{'} + (\bm{W}_2^k)^{*} \bm{\tilde{X}}^{k})$, where $(\theta^k)^{*}, (\bm{W}_1^k)^{*}$, and $(\bm{W}_2^k)^{*} $ are the learned parameters of the $k$-th layer.
\begin{remark}
Most of the existing deep learning based approaches including the proposed LISTA variations rely on the assumption that the channel coefficients follow the same distribution during the training and testing stages, and may not work well in dynamic environment with different channel distributions over time.
The continual learning and transfer learning technologies \cite{sun2021continus, yuan2021transfer} that have been developed for resource allocation and beamforming optimization in dynamic environment have the potential to address this issue. 
As these technologies are still in the early stage, we leave the proposed LISTA variations to recover the group-sparse matrix in dynamic environment for future work.
\end{remark}

\subsection{Time Complexity and Number of Trainable Parameters}
%The JADCE algorithms should hold low computational complexity to support massive access.
The time complexity for the three neural network structures are mainly due to matrix multiplication.
For LISTA-GS, the evaluation of the matrix multiplication $\bm{W}_1^k \bm{\tilde{Y}}$ and $\bm{W}_2^k \bm{\tilde{X}}^{k}$ require $\mathcal{O}(NLM + N^2M)$ time at each iteration.
As for LISTA-GSCP and ALISTA-GS, the evaluation of the matrix multiplication $(\bm{W}^k)^T (\bm{\tilde{Y}} - \bm{\tilde{S}} \bm{\tilde{X}}^k)$ and $\gamma^k \bm{W}^T (\bm{\tilde{Y}} - \bm{\tilde{S}} \bm{\tilde{X}}^k)$ require $\mathcal{O}(LNM)$ time at each iteration.
Since the proposed methods converge faster than ISTA and with the same complexity per iteration, thereby reducing the computational cost.

For $K$-layer RNN,
the ALISTA-GS contains only $2K$ total trainable parameters $\{  \theta^k ,  \gamma^k   \}_{k=0}^{K-1}$,
while LISTA-GS and LISTA-GSCP require $K(N^2+LN+1)$ variables
$\{ \bm{W}_1^k,\bm{W}_2^k, \theta^k \}_{k=0}^{K-1}$ and  $ K(LN+1)$ variables $\{ \bm{W}^k,  \theta^k  \}_{k=0}^{K-1}$,  respectively.
%The total number of free variables for these two frameworks are $K(N^2+LN+1)$ and $ K(LN+1)$, respectively.
We summarize the initialization and the required numbers of trainable parameters for the $K$-layer networks in Table \ref{ta:1}.
Since $\theta^k$ and $\gamma^k$ should be initialize to proper constants, we initialize $\theta^k = 0.1$ and $\gamma^k=1$ in this paper.
\begin{table*}[!t]
        \centering
        \caption{Number of trainable parameters(params) and initialization in the $K$- layer RNN.}
        \begin{tabular}{c|c|c|c}
                \toprule
                Network&Trainable params&Initialization&Number of params\\
                \midrule  
                LISTA-GS&$\{ \bm{W}_1^k,\bm{W}_2^k, \theta^k \}_{k=0}^{K-1}$&$ \bm{W}_1^k = \frac{1}{C}\bm{\tilde{S}}^T$, $\bm{W}_2^k = \bm{I} -  \frac{1}{C}\bm{\tilde{S}}^T \bm{\tilde{S}}, \theta^k = 0.1$&$K(N^2 + LN + 1)$ \\
                \midrule 
                LISTA-GSCP&$ \{ \bm{W}^k,  \theta^k  \}_{k=0}^{K-1}$&$ \bm{W}^k = \frac{1}{C}\bm{\tilde{S}}^T,\theta^k = 0.1$&$ K(LN+1)$ \\
                \midrule 
                ALISTA-GS&$\{ \theta^k ,  \gamma^k   \}_{k=0}^{K-1}$&$\theta^k = 0.1, \gamma^k =1.0 $  &$2K$ \\
                \bottomrule
        \end{tabular}
        \label{ta:1}
\end{table*}

\section{Convergence Analysis}\label{sec:con}
In this section, we provide the main theoretical results of this paper, i.e., 
the linear convergence rate of LISTA-GSCP \eqref{eq:LISTA-CP} and ALISTA-GS \eqref{eq:ALISTA}, respectively.
Since the proposed unrolled networks inherit the structure of ISTA-GS, they thus allow us to track the interpretability for such deep learning framework from the perspective of optimization.
As a matter of fact, our proposed unrolled neural networks are extensions of \cite{chen2018theoretical, liu2018alista} to solve the group-sparse matrix estimation problems.
Such dimension expansion and matrix structures bring the unique and formidable challenges for establishing the theoretical analysis on the proposed unrolled neural networks.

We firstly establish the convergence analysis of LISTA-GSCP framework.
We use $\bm{\tilde  X}^{k}$ to replace $\bm{\tilde  X}^{k}(\bm{\tilde  X}^{\natural}, \bm{\tilde{Z}})$ for notational simplicity.
The following theorem presents the convergence rate of LISTA-GSCP.
For theoretical analysis, we assume that $\ell_2$ norm of all rows of signal $\bm{\tilde  X}^{\natural}$ and Frobenius norm of AWGN noise $\bm{\tilde  Z}$ are bounded by $\beta$ and $\sigma$ \cite{chen2018theoretical, liu2018alista}.
Furthermore, since each entry of the activity sequence $ \{a_1,...,a_N\}$ follows the Bernoulli distribution, we assume that the number of non-zero rows on signal $\bm{\tilde  X}^{\natural}$ is bounded by a small number $s$ \cite{jiang2018joint}.
And for notation brevity, we assume the signal $\bm{\tilde  X}^{\natural}$ and noise $\bm{\tilde  Z}$ belong to set $\mathcal{X} (\beta, s, \sigma) 
        \coloneqq  \big\{ (\bm{\tilde  X}^{\natural}, \bm{\tilde{Z}} )|  \| \bm{\tilde{X}}^{\natural}[i,:] \|_2 \leq \beta, \forall i,
        \| \bm{\tilde  X}^{\natural}\|_{2,0} \leq s, \notag \|\bm{\tilde{Z}} \|_{F} \leq \sigma \big\}.$

\begin{theorem}[Convergence rate of LISTA-GSCP]\label{thm:convergence1}
        Given $\{\bm{W}^k, \theta^{k} \}_{k=0}^{\infty}$,  let $\{\bm{\tilde  X}^{k}\}_{k=1}^{\infty}$
        be generated by LISTA-GSCP in \eqref{eq:LISTA-CP} with an input $\bm{\tilde{Y}}$ observed by \eqref{mmv} and initial point $\bm{\tilde  X}^{0}=\bm{0}$. If $s$ is sufficiently small, then for all $(\bm{\tilde  X}^{\natural}, \bm{\tilde{Z}}) \in \mathcal{X}(\beta, s, \sigma),$ we have the error bound: %for all $k\in \mathbb{N}$
        \begin{equation}\label{eq:conrate1}
        \|\bm{\tilde  X}^{k}-\bm{\tilde  X}^{\natural}\|_{F} \leq s \beta \exp (-c k)+ C \sigma,
        \end{equation}
        where $c>0$ and $C>0$ are constants that depend only on $\bm{\tilde{S}}$ and $s$.
        %Recall that $s$, , and $\sigma$ is the noise level defined in \eqref{eq:assumption1}.
\end{theorem}
%Theorem \ref{thm:convergence1} presents an upper bound of $\|\bm{\tilde  X}^{k}-\bm{\tilde  X}^{\natural}\|_{F} $ for LISTA-GSCP.
Especially, for the noiseless case, \eqref{eq:conrate1} reduces to
$
        \|\bm{\tilde  X}^{k}-\bm{\tilde  X}^{\natural}\|_{F} \leq s \beta \exp (-c k).
$
Moreover, LISTA-GSCP converges at an $ \mathcal{O}(\log(\frac{1}{\epsilon}))$ rate, which is faster than original ISTA of $\mathcal{O}(\frac{1}{\epsilon})$ and Nesterov's method \cite{liu2009multi} of $\mathcal{O}(\frac{1}{\sqrt{\epsilon}})$.

The steps of proving Theorem \ref{thm:convergence1} are summarized as follows:
\begin{itemize}
        \item[A.] \textbf{\textquotedblleft Good" parameters for leaning.}
        We define the \textquotedblleft good" parameters that guarantee the linear convergence rate.
        \item[B.] \textbf{Error bound for one sample data.} We establish an error bound for one sample $(\bm{\tilde {X}}^{\natural},\bm{\tilde{Z}})$\footnote{To emphasize the signal and noise, we use $(\bm{\tilde {X}}^{\natural},\bm{\tilde{Z}})$ to replace one sample rather than $\{ \bm{\tilde{X}}^{\natural}, \bm{\tilde{Y}}\} $.}.
        \item[C.]  \textbf{Error bound for the whole data set.} By taking the supremum over all $(\bm{\tilde {X}}^{\natural},\bm{\tilde{Z}})$, we establish an error bound over the whole samples.
\end{itemize}
\subsection{\textquotedblleft Good" Parameters for Learning}
%The deep learning based ISTA-GS achieves a better performance by freeing the parameters to be learned.
In this subsection, we shall give the definition of \textquotedblleft good" parameters for the LISTA-GSCP network structure that guarantees the linear convergence rate.% that guarantee of Theorem \ref{thm:convergence1}. 

First, we need to introduce several fundamental definitions inspired by \cite{chen2018theoretical}.
The first one is the mutual coherence \cite{gribonval2003sparse} of $\bm{\tilde{S}}$, which characterizes the coherence between different columns of $\bm{\tilde{S}}$.

\begin{definition}
\begin{itemize}
        \item[(i)] The mutual coherence of $\bm{\tilde{S}} \in \R^{2L \times 2N}$ with normalized columns is defined as
        \begin{equation}
        \mu(\bm{\tilde{S}})=\max _{i \neq j \atop 1 \leq i, j \leq 2N} \big|(\bm{\tilde{S}}[:,i])^{T}\bm{\tilde{S}}[:,j] \big| .
        \end{equation}
        
        \item[(ii)] The generalized mutual coherence of $\bm{\tilde{S}} \in \R^{2L \times 2N}$ with normalized columns is defined as
        \begin{equation}\label{eq:g-mutual-coherence}
        \tilde{\mu}(\bm{\tilde{S}})=\inf_{\bm{W} \in \R^{2L \times 2N} \atop (\bm W[:,i])^{T} \bm{\tilde{S}}[:,i]=1,\forall i} \bigg\{ \max _{i \neq j \atop 1 \leq i, j \leq 2N} \Big|(\bm W[:,i])^{T} \bm{\tilde{S}}[:,j] \Big| \bigg\}.
        \end{equation}
\end{itemize}
\end{definition}

Lemma $1$ in \cite{chen2018theoretical} tells us that there exists a matrix $ \bm{W} \in \R^{2L \times 2N}$ that attaches the infimum given in \eqref{eq:g-mutual-coherence}, i.e., $\mathcal{W}(\bm{\tilde{S}}) \neq \emptyset,$
where a set of \textquotedblleft good" weight matrices is defined as
\begin{align}\label{chwei}
         \mathcal{W}(\bm{\tilde{S}}) \coloneqq \argmin_{\bm{W} \in \R^{2L \times 2N}} \big\{ \| \bm{W}\|_{\max} \big| 
         (\bm W[:,i])^{T} \bm{\tilde{S}}[:,i]=1, \forall i,  \notag \\
         \max _{i \neq j \atop 1 \leq i, j \leq 2N} \big|(\bm W[:,i])^{T} \bm{\tilde{S}}[:,j]\big| =\tilde{\mu}(\bm{\tilde{S}})     
          \big\}. 
\end{align}

Then, we define the \textquotedblleft good" parameters to be learned in LISTA-GSCP as follows.

\begin{definition}\label{def:good1}
$\bm{\Theta}= \{\bm{W}^k, \theta^{k}\}_{k=0}^{\infty}$ are called \textquotedblleft good" parameters in LISTA-GSCP if they satisfy
%for all  $(\bm{\tilde X}^{\natural},\bm{\tilde{Z}}) \in \mathcal{X}(\beta,s,\sigma)$ 
        \begin{equation}\label{g1}
                \bm{W}^k \in \mathcal{W}(\bm{\tilde{S}}), \quad
                \theta^{k}=\tilde{\mu} \sup_{(\bm{\tilde{X}}^{\natural}, \bm{\tilde{Z}} )} \|\bm{\tilde{X}}^{k} - \bm{\tilde{X}}^{\natural}\|_{2,1}+ \sigma C_W, \quad \forall k\in \mathbb{N},
        \end{equation}
        where  $C_W = \underset{k\geq 0}{\max}~ \| \bm{W}^k \|_{2,1}$ and $\tilde{\mu}=\tilde{\mu} (\bm{\tilde{S}})$.
%       where 
%$C_{W}=\max _{1 \leq i, j \leq 2N}|\bm{W}^k_{i, j}|$ and  
%       $\tilde{\mu} = \tilde{\mu}(\bm{\tilde{S}})$.
\end{definition}

In Definition \ref{def:good1}, we propose the \textquotedblleft good" choice of the learning parameters in LISTA-GSCP.
In the following subsection, we prove that the sequence of \textquotedblleft good" parameters lead to the conclusion \eqref{eq:conrate1} in Theorem \ref{thm:convergence1}.

\subsection{Error Bound for One Sample} 
In this subsection, we give an upper bound of the recovery error for one sample $(\bm{\tilde {X}}^{\natural},\bm{\tilde{Z}}) \in  \mathcal{X} (\beta, s, \sigma) $.
We first introduce the extra notation $\psi$, to provide information about the group sparsity.
For each $\bm{\tilde{X}} \in \mathbb{R}^{2N \times M}$, we define a function $\psi:\mathbb{R}^{2N \times M} \rightarrow \mathbb{R}^{2N}$ as 
\begin{equation}
        \psi(\bm{\tilde{X}}) = 
        \begin{bmatrix}
        ||\bm{\tilde{X}}[1,:] ||_2,
        ||\bm{\tilde{X}}[2,:]||_2,
        \cdots,
        ||\bm{\tilde{X}}[2N,:] ||_2
        \end{bmatrix}^T .
\end{equation}
Specifically, for a vector $\bm{v}=[v_1,v_2,\ldots,v_{2N}]^T \in \mathbb{R}^{2N}$, we have
$
\psi(\bm{v})_i = |v_i| , \text{for all }i\in [2N].
$
Note that $\supp(\psi(\bm{A}))$ can provide information about group sparsity of a given matrix $\bm{A}$. 
By simple calculations, one can get the following lemma.
\begin{lemma}
        With  $\bm{\tilde{X}} \in \mathbb{R}^{2N \times M}$ and $\bm{e}=[1,1,\ldots, 1]^T$, we have
        \begin{equation}
                        \psi( \eta_{\theta}(\bm{\tilde{X}}) ) = \psi (\bm{\tilde{X}})  - \theta^k \bm{e},
        \end{equation}
                        %where $\bm{e}=[1,1,\ldots, 1]^T$.
\end{lemma}

By taking one sample $(\bm{\tilde X}^{\natural},\bm{\tilde{Z}}) \in \mathcal{X}(\beta,s,\sigma)$ and letting $\mathcal{I} = \supp( \psi( \bm{\tilde X}^{\natural}) )$, we establish the error bound by two steps: (i) We show that there are no false positive rows in $ \bm{\tilde{X}}^k$ for all $k$.
%, i.e., $\bm{\tilde{X}}^{k}[i,:] =\bm{0}$ (equivalently, $\| \bm{\tilde{X}}^{k}[i,:] \|_2=0$), $ \forall i \notin \mathcal{I}, ~\forall k.$
(ii) Since the no-false-positive property holds, we consider the component on $\mathcal{I}$.

In step (i), we prove the following lemma.
\begin{lemma}[No-false-positive property]\label{le:nofp}
        With all assumptions in Theorem \ref{thm:convergence1} and \textquotedblleft good" parameters $\bm{\Theta}$, we have
        \begin{equation}
                \supp(\psi(\bm{\tilde {X}}^{k})) \subset \supp( \psi( \bm{\tilde X}^{\natural}) ), \quad \forall k.
        \end{equation}
\end{lemma}
\begin{proof}
        Please refer to Appendix \ref{pr:le1}.
\end{proof}
Lemma \ref{le:nofp} shows that there are no false positive entries in $\bm{\tilde {X}}^{k}$.
In another word, if \eqref{g1} in LISTA-GSCP hold, then
\begin{equation}\label{eq:nofalsepositive1}
         \bm{\tilde{X}}^{k}[i,:] =\bm{0}, \quad \forall i \notin \mathcal{I}, ~\forall k.
\end{equation}
This property implies that the recovery error of the component beyond $\mathcal{I}$ turns to be $0$.

In step (ii), we consider the component on $\mathcal{I}$.
For all $i \in \mathcal{I}$, the LISTA-GSCP in \eqref{eq:LISTA-CP} gives
\begin{align}
        &\bm{\tilde X}^{k+1}[i,:]
        = \eta_{\theta^k} \Big( \bm{\tilde{X}}^k[i,:]  - (\bm{W}^k[:,i])^T \bm{\tilde{S}}[:,\mathcal{I}]  \nonumber (  \bm{\tilde{X}}^k[\mathcal{I},:] - \\
         &\bm{\tilde{X}}^{\natural}[\mathcal{I},:]) +  (\bm{W}^k[:,i])^T \bm{\tilde{Z}}   \Big)  \nonumber \\
        & \in \underbrace{  \bm{\tilde{X}}^k[i,:]  - (\bm{W}^k[:,i])^T \bm{\tilde{S}}[:,\mathcal{I}]  
         (  \bm{\tilde{X}}^k[\mathcal{I},:] -   \bm{\tilde{X}}^{\natural}[\mathcal{I},:]) }_{(T1)}   \notag \\
         &+ \underbrace{  (\bm{W}^k[:,i])^T\bm{\tilde{Z}}   }_{(T2)}
          - \underbrace{  \theta^k \partial \Vert\bm{\tilde X}^{k+1}[i,:]\Vert_{2} }_{(T3)},
          \label{eq:subgradient21}
\end{align}
where $\partial \|\cdot \|_{2}$ is given in \eqref{l2norm}.
It can be viewed that \eqref{eq:subgradient21} consists of three parts $(T1)$, $(T2)$, $(T3)$.
%We first consider part $(T1)$. 
%$$\bm{\tilde{X}}^k[i,:]  - (\bm{W}^k[:,i])^T \bm{\tilde{S}}[:,\mathcal{I}] (\bm{\tilde{X}}^k[\mathcal{I},:]  - \bm{\tilde{X}}^{\natural}[\mathcal{I},:]).$$
Since $(\bm{W}^k[:,i])^T \bm{\tilde{S}}[:,i] = 1$, 
%Here, in (a) we separate index $i$ from $\mathcal{I}$ and  we use the fact that $(\bm{W}^k[:,i])^T \bm{\tilde{S}}[:,i] = 1$ as the weight matrix obeys \eqref{g1},
%(b) occurs when we simplify the equation.
then \eqref{eq:subgradient21} can expressed as 
% we put $\bm{\tilde X}^{k+1}[i,:]$ and  $\bm{\tilde{X}}^{\natural}[i,:]$ to the left side:
\begin{align}
        &\bm{\tilde X}^{k+1}[i,:] - \bm{\tilde{X}}^{\natural}[i,:] 
        \in - \sum_{j\in \mathcal{I}, j \ne i} (\bm{W}^k[:,i])^T \bm{\tilde{S}}[:,j] \notag \\
        & \Big(\bm{\tilde{X}}^k[j,:] - \bm{\tilde{X}}^{\natural}[j,:] \Big) 
        + (T2) - (T3). \label{eq:subgradient3}
\end{align}

The definition of $\partial \|\cdot \|_{2}$ in \eqref{l2norm} shows that 
$\| \partial \Vert\bm{\tilde X}^{k+1}[i,:]\Vert_{2}\|_2 \leq 1.$ 
Then, taking the norm on both sides in \eqref{eq:subgradient3}, one can get that for all $i \in \mathcal{I}$,
\begin{align}\label{mid1}
        \big\| \bm{\tilde X}^{k+1}[i,:] - \bm{\tilde{X}}^{\natural}[i,:]  \big\|_2 \overset{\text{(a)}}{\leq}  \sum_{j\in \mathcal{I}, j \ne i} \big|(\bm{W}^k[:,i])^T \bm{\tilde{S}}[:,j]  \big| \cdot \notag \\
         \big\| \bm{\tilde{X}}^k[j,:]  - \bm{\tilde{X}}^{\natural}[j,:] \big\|_2 + \theta^k + \big\Vert (\bm{W}^k[:,i])^T \bm{\tilde{Z}} \big\Vert_2  \notag\\
        %&\overset{\text{(b)}}{\leq} \tilde{\mu} \sum_{j\in \mathcal{I}, j \ne i} \big\| \bm{\tilde{X}}^k[j,:]  - \bm{\tilde{X}}^{\natural}[j,:]  \big\|_2 + \theta^k + \big\Vert (\bm{W}^k[:,i])^T \bm{\tilde{Z}} \big\Vert_2 \nonumber \\
        \overset{\text{(b)}}{\leq} \tilde{\mu} \sum_{j\in \mathcal{I}, j \ne i} \big\| \bm{\tilde{X}}^k[j,:]  - \bm{\tilde{X}}^{\natural}[j,:]  \big\|_2
         + \theta^k + \big\| \bm{W}^k[:,i] \big\|_2  \big\| \bm{\tilde{Z}} \big\|_F,
\end{align}
where (a) follows from the triangle inequality, (b) arises from the choice of \textquotedblleft good" parameters in \eqref{def:good1} and 
$
\big\| (\bm{W}^k[:,i])^T \bm{\tilde{Z}}  \big\|_2
\leq  \big\| \bm{W}^k[:,i] \big\|_2  \| \bm{\tilde{Z}} \|_F. 
$
It is easy to check that \eqref{eq:nofalsepositive1} implies that $ \big\Vert \bm{\tilde{X}}^k - \bm{\tilde{X}}^{\natural} \big\Vert_{2,1} = \big\Vert \bm{\tilde{X}}^k[\mathcal{I},:] -   \bm{\tilde{X}}^{\natural}[\mathcal{I},:] \big\Vert_{2,1} $ for all $k$. 
Therefore, based on \eqref{mid1}, it follows that 
\begin{align}\label{mid2}
        &\Vert \bm{\tilde{X}}^{k+1} - \bm{\tilde{X}}^{\natural} \Vert_{2,1} = \sum_{i \in \mathcal{I}} \Vert \bm{\tilde{X}}^{k+1}[i,:] -   \bm{\tilde{X}}^{\natural}[i,:] \Vert_{2} \notag \\
        %&\leq \sum_{i \in \mathcal{I}} \Big( \tilde{\mu}  \sum_{j\in \mathcal{I}, j \ne i} \big\Vert(\bm{\tilde{X}}^k[j,:]  - \bm{\tilde{X}}^{\natural}[j,:]) \big\Vert_2 \nonumber \\
        %&\qquad \qquad  + \theta^k + \big\| \bm{W}^k[:,i] \big\|_2   \| \bm{\tilde{Z}} \|_F   \Big)  \notag \\
        %&=(|\mathcal{I}| - 1) \tilde{\mu} \sum_{i \in \mathcal{I}} \big\Vert(\bm{\tilde{X}}^k[i,:]  - \bm{\tilde{X}}^{\natural}[i,:]) \big\Vert_2  \nonumber \\
        %& \qquad + |\mathcal{I}|\theta^k +   \| \bm{\tilde{Z}} \|_F \big\| \bm{W}^k \big\|_{2,1}   \notag \\
        &\leq \tilde{\mu}(|\mathcal{I}| - 1) \big\Vert \bm{\tilde{X}}^k - \bm{\tilde{X}}^{\natural} \big\Vert_{2,1}  + |\mathcal{I}|\theta^k +   \sigma C_W , 
\end{align}
where $C_W = \underset{k\geq 0}{\max}~ \| \bm{W}^k \|_{2,1}$.
%Here, (e) arises from \eqref{mid1}. 
The inequality \eqref{mid2} provides the recursive form for consecutive errors of $\Vert \bm{\tilde{X}}^{k+1} - \bm{\tilde{X}}^{\natural} \Vert_{2,1}$ and $\Vert \bm{\tilde{X}}^{k} - \bm{\tilde{X}}^{\natural} \Vert_{2,1}$.
Hence, we establish the recover error bound of mixed-norm for one sample $(\bm{\tilde {X}}^{\natural},\bm{\tilde{Z}}) \in \mathcal{X}(\beta, s, \sigma)$.

\subsection{Error Bound for Whole Data Set}
In this subsection, we give an upper bound of recovery error for the whole training samples.
Note that  $|\mathcal{I}|:= \supp( \psi( \bm{\tilde X}^{\natural}) ) \leq s$ for all $(\bm{\tilde{X}}^{\natural},\bm{\tilde{Z}})\in\mathcal{X}(\beta, s, \sigma)$.
Taking supremum over $(\bm{\tilde{X}}^{\natural},\bm{\tilde{Z}})\in\mathcal{X}(\beta, s, \sigma)$ on both sides of \eqref{mid2}, we have
\begin{align}
        \sup_{(\bm{\tilde X}^{\natural}, \bm{\tilde{Z}})}  \Vert \bm{\tilde{X}}^{k+1} - \bm{\tilde{X}}^{\natural} \Vert_{2,1}   \leq (s - 1) \tilde{\mu} \sup _{(\bm{\tilde X}^{\natural}, \bm{\tilde{Z}})}  \Vert \bm{\tilde{X}}^k - \bm{\tilde{X}}^{\natural} \Vert_{2,1} \notag \\
         + s\theta^k +\sigma  C_W.
\end{align}
Considering the \textquotedblleft good" parameters of 
$\theta^{k}=\tilde{\mu} \sup_{(\bm{\tilde{X}}^{\natural}, \bm{\tilde{Z}}) } \|\bm{\tilde{X}}^{k}-\bm{\tilde{X}}^{\natural}\|_{2,1} +\sigma C_{W} ,$
it follows that
\begin{align*}
        \sup _{(\bm{\tilde X}^{\natural}, \bm{\tilde{Z}})}  \Vert \bm{\tilde{X}}^{k+1} - \bm{\tilde{X}}^{\natural} \Vert_{2,1}   
        %&\leq (2\tilde{\mu}s - \tilde{\mu})  \sup_{(\bm{\tilde X}^{\natural}, \bm{\tilde{Z}})}  \Vert \bm{\tilde{X}}^k - \bm{\tilde{X}}^{\natural} \Vert_{2,1}  + (s+1) \sigma   C_W\\
        %&\leq (2\tilde{\mu}s - \tilde{\mu})^2 \sup _{(\bm{\tilde X}^{\natural}, \bm{\tilde{Z}})}  \Vert \bm{\tilde{X}}^{k-1} - \bm{\tilde{X}}^{\natural} \Vert_{2,1}  \\
        %& \qquad+  ( 1 + 2\tilde{\mu}s - \tilde{\mu} )(s+1) \sigma C_W  \\
        %& \qquad \qquad \vdots \\
        %&\leq (2\tilde{\mu}s - \tilde{\mu})^{k+1}\sup _{(\bm{\tilde X}^{\natural}, \bm{\tilde{Z}})}  \Vert \bm{\tilde{X}}^{0} - \bm{\tilde{X}}^{\natural} \Vert_{2,1} \\
        %& \qquad + (s+1)  \sigma C_W  \sum_{i=0}^{k+1} (2\tilde{\mu}s - \tilde{\mu})^{i}  \\    
        \leq (2\tilde{\mu}s - \tilde{\mu})^{k+1} s \beta + \frac{(s+1) C_W}{1+\tilde{\mu}-2\tilde{\mu}s} \sigma,
\end{align*}
provided that $2\tilde{\mu}s - \tilde{\mu} <1$.
By letting $c=-\log (2 \tilde{\mu} s-\tilde{\mu}), C=\frac{(s+1) C_{W}}{1+\tilde{\mu}-2 \tilde{\mu} s}$, we have 
$
        \sup _{(\bm{\tilde X}^{\natural}, \bm{\tilde{Z}})} \Vert \bm{\tilde{X}}^{k+1} - \bm{\tilde{X}}^{\natural} \Vert_{2,1}   \leq s \beta \exp(-c(k+1)) +  C \sigma .
$
The fact that $ \| \bm{X} \|_{F} \leq \| \bm{X} \|_{2,1}$ for any matrix $ \bm{X}$ gives an upper bound of error with respect to Frobenius norm
\begin{align}
        \sup _{(\bm{\tilde X}^{\natural}, \bm{\tilde{Z}})}  \Vert \bm{\tilde{X}}^{k+1} - \bm{\tilde{X}}^{\natural} \Vert_{F}   &\leq \sup _{(\bm{\tilde X}^{\natural}, \bm{\tilde{Z}})} \Vert \bm{\tilde{X}}^{k+1} - \bm{\tilde{X}}^{\natural} \Vert_{2,1} \notag \\
        &\leq s \beta \exp(-c(k+1)) + \sigma C .
\end{align}
As long as  $s\le(1+1 / \tilde{\mu}) / 2$ and $ c=-\log (2 \tilde{\mu} s-\tilde{\mu})>0$ hold, the error bound holds uniformly
for all $(\bm{\tilde X}^{\natural}, \bm{\tilde{Z}}) \in \mathcal{X}(\beta, s, \sigma)$.
Then for $k$-layer network, we have
\begin{align}
\sup _{(\bm{\tilde X}^{\natural}, \bm{\tilde{Z}})}  \Vert \bm{\tilde{X}}^{k} - \bm{\tilde{X}}^{\natural} \Vert_{F}  \leq s \beta \exp(-ck) + \sigma C .
\end{align}
Therefore, we complete the proof of Theorem \ref{thm:convergence1}.%, which shows the linear convergence rate of LISTA-GSCP.

\subsection{Convergence Rate of ALISTA-GS}
In this subsection, we first give the definition of the \textquotedblleft good" parameters to be learned in the ALISTA-GS network.
We then verify the convergence of ALISTA-GS for noiseless case.
For noisy case, the analysis can be referred to that of Theorem $1$.
\begin{definition}\label{def:good2}
 $\bm {\Theta} = \{  \theta^k , \gamma^k   \}_{k=0}^{\infty} $ are called \textquotedblleft good" parameters for all  $(\bm{\tilde X}^{\natural}, \bm{\tilde{Z}}) \in \mathcal{X}(\beta,s,0)$ (noiseless case) with the weight matrix is pre-computed in ALISTA-GS if they satisfy that for each $k\in \mathbb{N}$
        \begin{equation}\label{g3}
        \bm{W} \in \mathcal{W}(\bm{\tilde{S}}), \quad
        \theta^{k}=\tilde{\mu} \gamma^k \sup _{(\bm{\tilde X}^{\natural}, \bm{\tilde{Z}})} \|\bm{\tilde{X}}^{k}-\bm{\tilde{X}}^{\natural}\|_{2,1}.
        \end{equation}
        %where 
        %$C_{W}=\max _{1 \leq i, j \leq 2N}|\bm{W}_{i, j}|$ and $\tilde{\mu} = \tilde{\mu}(\bm{\tilde{S}})$.
\end{definition}

\begin{theorem}[Convergence rate of ALISTA-GS]\label{thm:convergence2}
        Given $\{ \theta^k, \gamma^k \}_{k=0}^{\infty} $ and $\bm{W} \in \mathcal{W}(\bm{\tilde{S}})$,
        let $\{\bm{\tilde  X}^{k}\}_{k=1}^{\infty}$
        be generated by ALISTA-GS in \eqref{eq:ALISTA} with an input $\bm{\tilde{Y}}$ observed by \eqref{mmv} and initial point $\bm{\tilde  X}^{0}=\bm{0}$. If $s\le(1+1 / \tilde{\mu}) / 2$ for all $(\bm{\tilde  X}^{\natural}, \bm{\tilde{Z}}) \in \mathcal{X}(\beta, s, 0),$ and
        \begin{align}
                &\gamma^k \in \Big(0, \frac{2}{1+2\tilde{\mu} s - \tilde{\mu} }\Big), \notag \\
                &\theta^{k}=  \tilde{\mu} \gamma^k \sup_{  (\bm{\tilde  X}^{\natural}, \bm{\tilde{Z}})     }   \| \bm{\tilde{X}}^{k}-\bm{\tilde{X}}^{\natural} \|_{2,1} ,\forall k \in \mathbb{N} ,  \label{eq:thm3-condition2}
        \end{align}
        then we have
        \begin{align}
        \Vert \bm{\tilde{X}}^{k+1} - \bm{\tilde{X}}^{\natural} \Vert_{F}  
        \leq s \beta \exp\Big(-\sum_{\tau=0}^{k}c^{\tau} \Big),
        \end{align}
        where $c^{\tau} = -\log\big( \gamma^\tau(2 \tilde{\mu} s-\tilde{\mu}) + |1 - \gamma^\tau|\big)$ is a positive constant.
\end{theorem}

\begin{proof}
        Please refer to Appendix \ref{pr:th3}.
\end{proof}
\begin{remark}
        Optimally, if the factor $c^{\tau}$ takes the maximum at $\gamma^\tau = 1 $, i.e., $c^{\tau}  \equiv -\log( 2 \tilde{\mu} s-\tilde{\mu})$, then ALISTA-GS enjoys a linear convergence rate.
\end{remark}

\section{Numerical Results}\label{sec:num}
In this section, we present the simulation results for the proposed three unrolled neural network structures for grant-free massive access in IoT networks.
We first introduce the settings of the simulations and performance metrics and then present the simulation results to confirm the
main theorems including the weight coupling structure and the convergence rate.
%To meet the demands of massive device connectivity,
Finally, we compare the recovery performance of the proposed methods with other popular CS-based methods.
\begin{figure}[htb!]
        \centering
        \subfigure[Weight $ \bm{W}_{2}^{k}-(\bm I-\bm{W}_{1}^{k} \bm{\tilde{S}}) \rightarrow \bm{0}$.]{
                \includegraphics[width=0.40\linewidth]{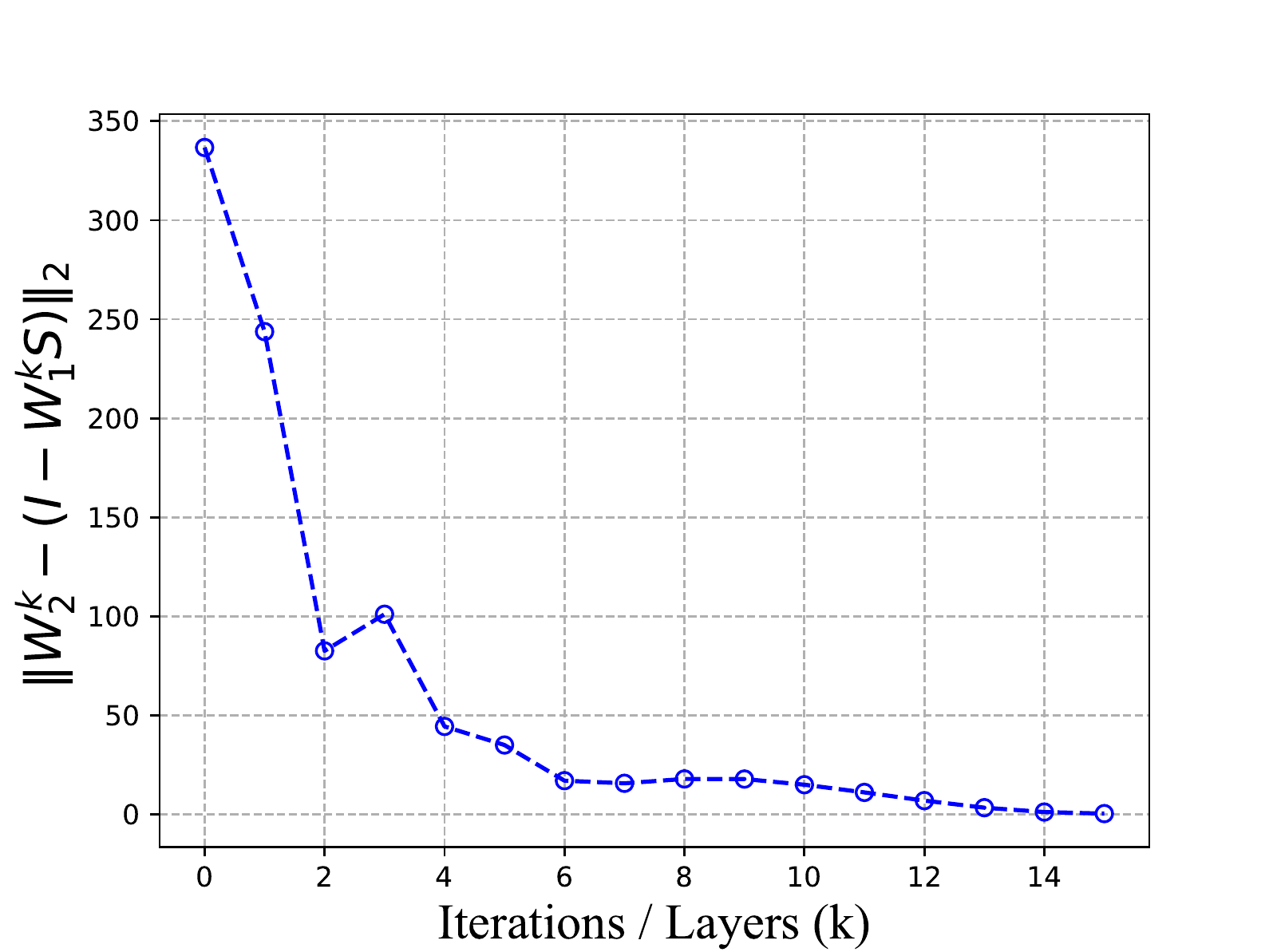}
                \label{fig:tha}
        }
        \subfigure[The threshold $ \theta^{k} \rightarrow 0$.]{
                \includegraphics[width=0.40\linewidth]{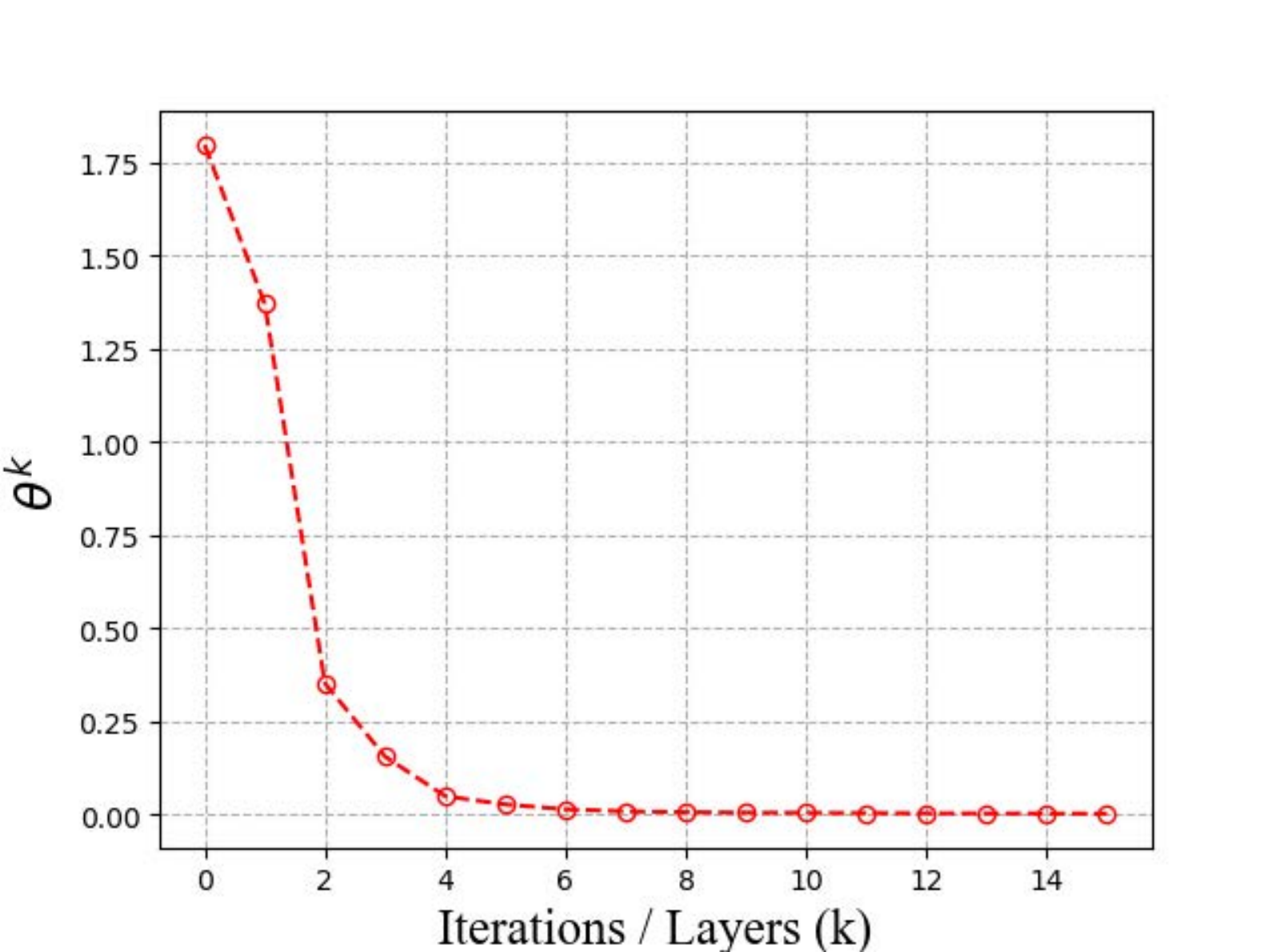}
                \label{fig:thb}}
        \caption{Validation of Theorem \ref{thm:main1}. LISTA-GS has weight coupling structure.} \label{fig:v1}
\end{figure}

\subsection{Simulation Settings and Performance Metrics}
In simulations, 
%we set the length of the signature sequence $L$, the total number of devices $N$, and the number of antennas at the BS $M$ to $100$, $200$, and $5$, respectively. 
%we generate the signature matrix according to the complex Gaussian distribution, i.e., $ \bm{S} \sim \mathcal{C}\mathcal{N}(\bm{0}, \bm{I})$.
the channels are assumed to suffer from independent Rayleigh fading, i.e., $\bm{H} \sim \mathcal{C}\mathcal{N}(\bm{0}, \bm{I}) $. 
The preamble signature matrix $\bm{S}$ is fixed with each of its columns being normalized and the noise matrix $\bm{Z}$ follows the Gaussian distribution with zero mean and variance $\sigma^2$.
In addition, each entry of the activity sequence $ \{a_1,...,a_N\}$ 
is a random variable which follows the Bernoulli distribution with mean $0.1$, i.e., $\mathbb{P}(a_n = 1) = 0.1$ and $\mathbb{P}(a_n = 0) = 0.9$, $\forall \, n \in [N] $. %Hence the activity matrix $ \bm{A} = \diagg(a_1,...,a_N)$ is generated and the group row sparse matrix is created by $ \bm{X} = \bm{A}\bm{H}$. 
The transmit signal-to-noise ratio (SNR) of the system is defined as
\begin{equation}
        \text{SNR} = \frac{\E[\| \bm{SX} \|_{F}^2]}{\E[\| \bm{Z} \|_{F}^2]} .
\end{equation}

By transforming all these complex-valued matrices into real-valued matrices according to (\ref{realmmv}), we obtain the training data set $  \{ \bm{\tilde{X}}_i^{\natural}, \bm{\tilde{Y}}_i\}_{i=1}^P$.
In the training stage, we choose $K = 12 $ layers unless otherwise stated for all the unrolled models in the simulations.
The training set contains $ P=64$ different samples in the training stage. 
The learning rate $\alpha_0$ is set to be $ 5 \times 10^{-4}$ and $\alpha_1 = 0.2\alpha_0$. 
In the validation and testing stage, $ 128$ samples are generated to test the trained models (drawn independent of the training set, but from the same distribution). 
%As for ISTA-GS, we set $\lambda = 0.2$, $0.1$, and $0.05$.  
%Furthermore, we initialize the parameters as $ W_1^0 = \frac{1}{C}\bm{\tilde{S}}$, $W_2^0 = \bm{I} -  \frac{1}{C}\bm{\tilde{S}}^T \bm{\tilde{S}}$, and $\theta = \frac{ 0.1}{C}$.

We compare our proposed unrolled networks, i.e., LISTA-GS, LISTA-GSCP and ALISTA-GS, with the other three popular CS-based algorithms for group-sparse matrix estimation problem:
\begin{itemize}
        \item \textbf{ISTA-GS} \cite{yuan2006model}: The vanilla ISTA for solving multiple measurement vector CS-based problem.
         The regularization parameter $\lambda$ is set as $0.1$.
         \item \textbf{Nesterov's method} \cite{liu2009multi}: A first-order method solves the equivalent smooth convex reformulations of \eqref{groupLASSO}.
         The regularization parameter $\lambda$ is set as $0.1$.
         \item \textbf{AMP-MMV} \cite{ziniel2012efficient, chen2018sparse}: The AMP developed Bayesian algorithm for solving the multiple measurement vector CS-based problem. 
         \item \textbf{AMP MMSE-denoiser} \cite{liu2018massive}: The vector AMP algorithm with minimum mean-squared error (MMSE) denoiser for user activity detection and channel estimation.
\end{itemize}

We adopt the normalized mean square error (NMSE) to evaluate the performance in recovering the real-valued $ \bm{\tilde{X}}$, defined as 
\begin{equation}
\text{NMSE}(\bm{\tilde{X}}, \bm{\tilde{X}}^{\natural}) = 10 \log_{10} \left(\frac{\mathbb{E}\Vert \bm{\tilde{X}} - \bm{\tilde{X}}^{\natural} \Vert_{F}^2}{\mathbb{E}\Vert \bm{\tilde{X}}^{\natural} \Vert_{F}^2}\right),
\end{equation}
where $\bm{\tilde{X}}^{\natural}$ represents the ground truth and $ \bm{\tilde{X}}$ is the estimated value.

\subsection{Validation of Theorems}
\begin{figure*}[t]
        \centering
        \begin{minipage}{.46\textwidth}
                \centering
                \includegraphics[width=1.0\columnwidth]{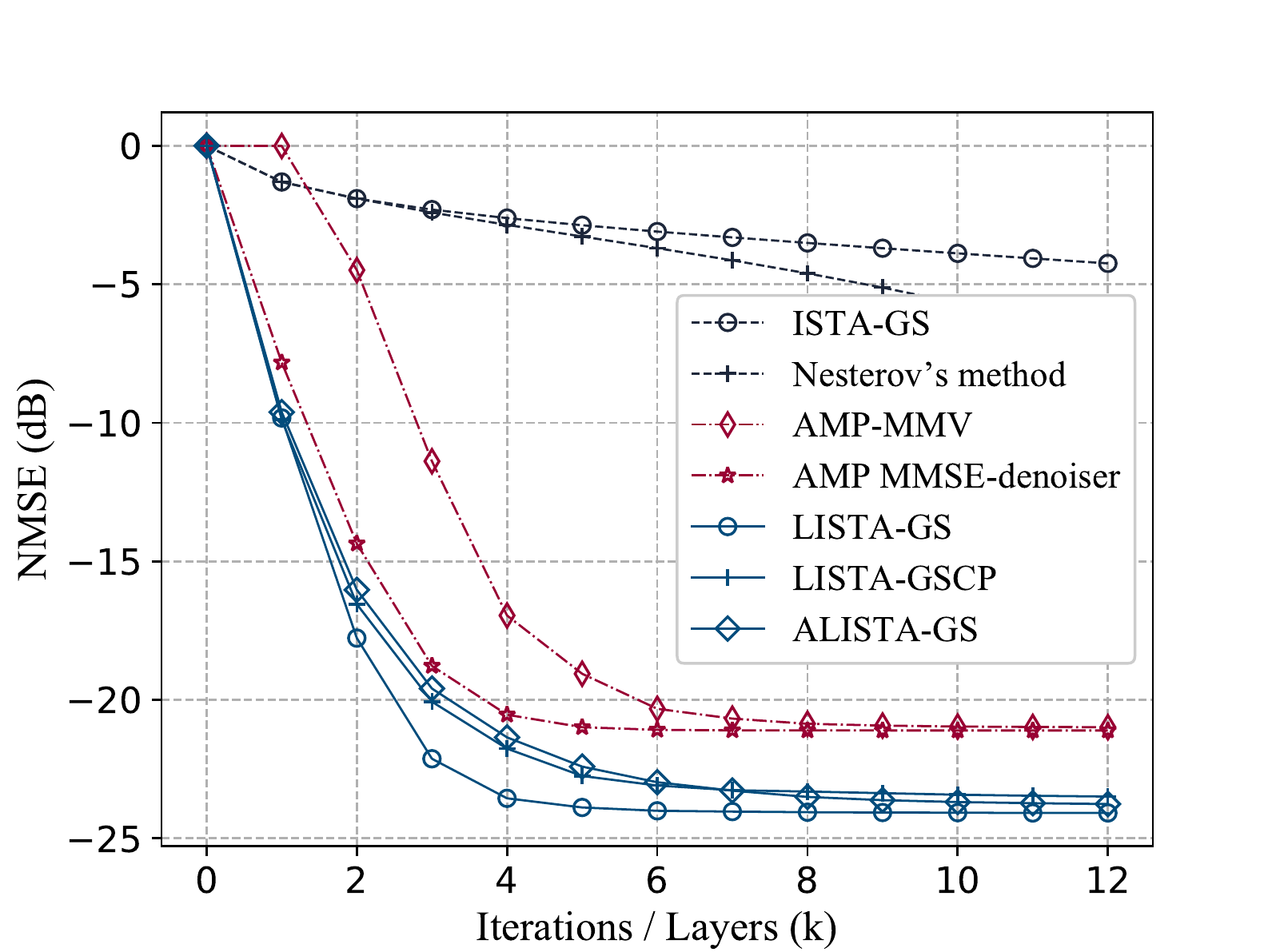}
                \caption{NMSE versus iterations with SNR $=15~$dB.}\label{fig:1}
        \end{minipage}
        \begin{minipage}{.46\textwidth}
                \centering
                \includegraphics[width=1.0\columnwidth]{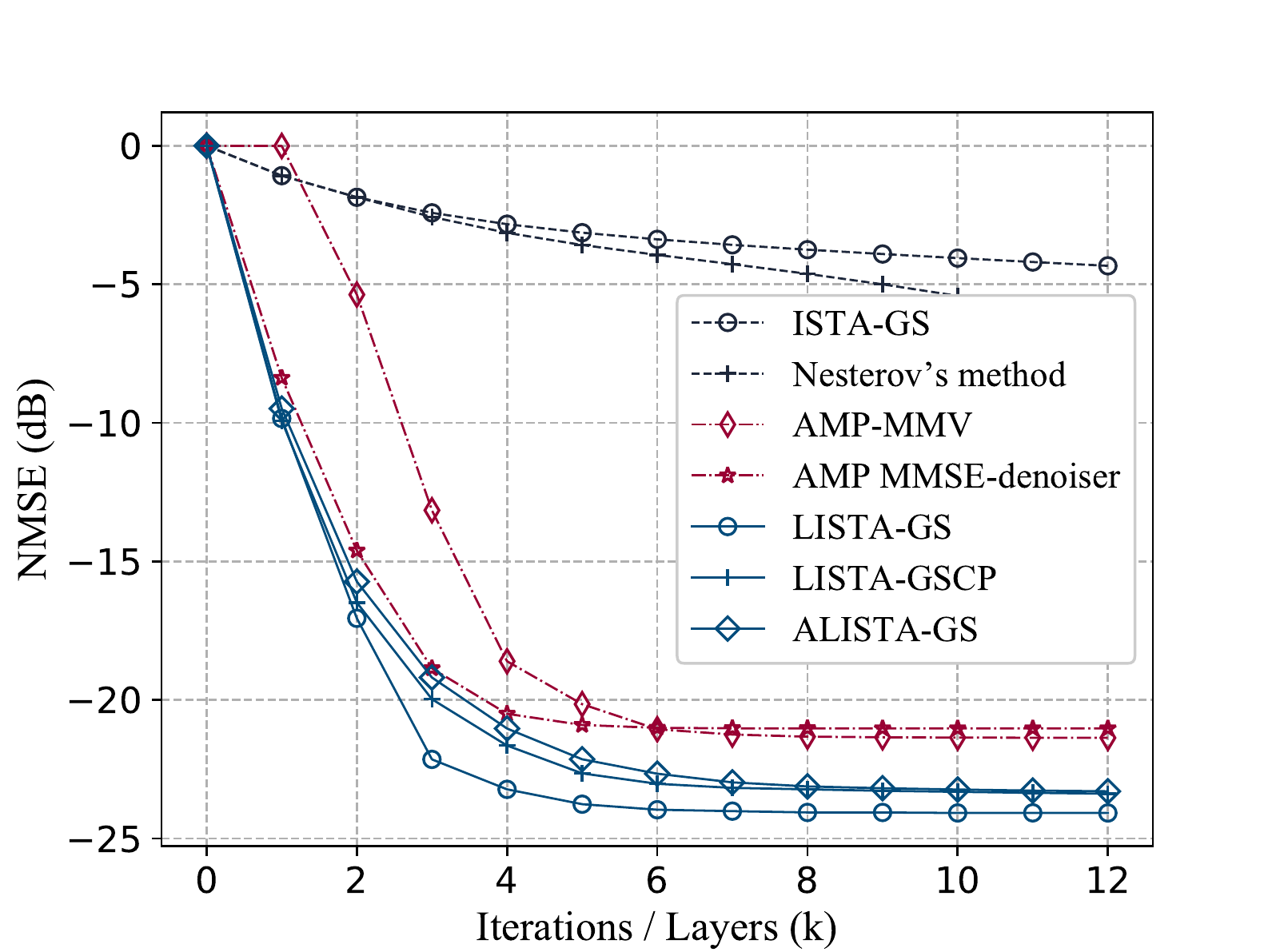}
                \caption{NMSE of unrolled networks and other methods when condition number $ \kappa = 2$ and SNR $=15~$dB. }\label{fig:k2}
        \end{minipage}

    \begin{minipage}{.46\textwidth}
        \centering
        \includegraphics[width=1.0\columnwidth]{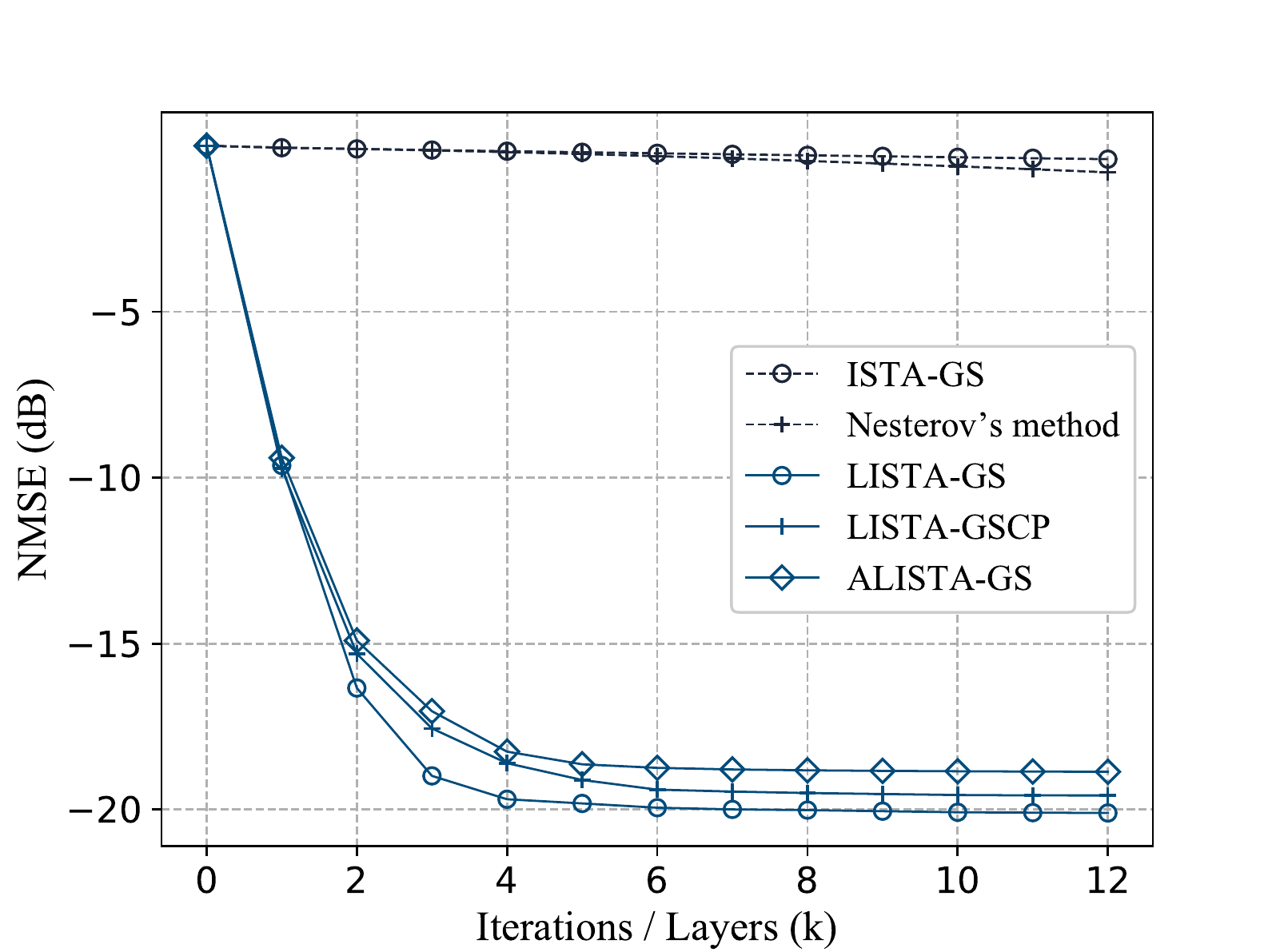}
        \caption{NMSE of urolled networks and other methods when condition number $ \kappa = 15$ and SNR $=15~$dB. AMP-MMV and AMP MMSE-denoiser fail in this case.}\label{fig:k50}
    \end{minipage}
    \begin{minipage}{.46\textwidth}
        \centering
        \includegraphics[width=1.0\columnwidth]{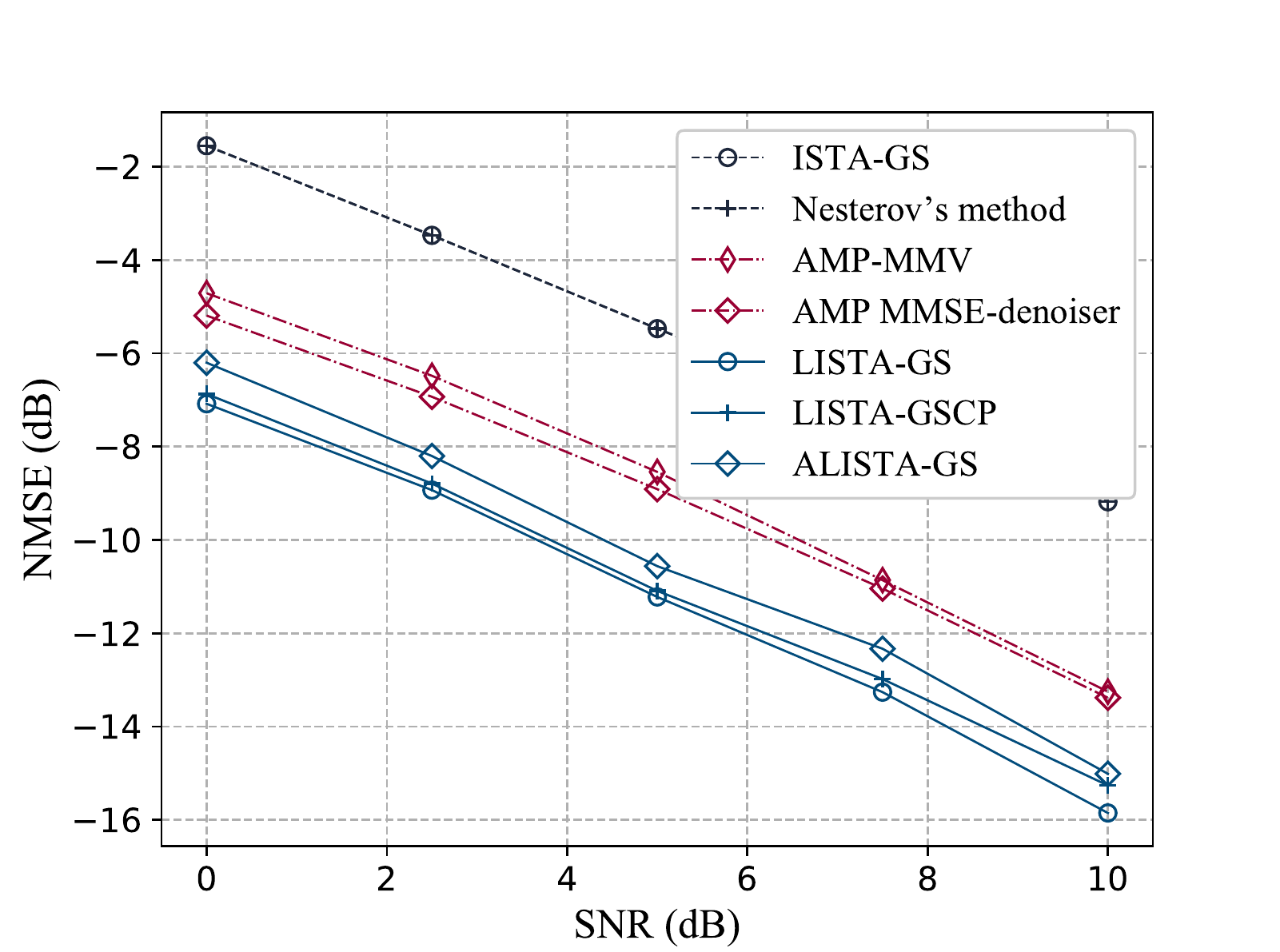}
        \caption{NMSE versus SNR when preamble signature $\bm{S}$ is complex Gaussian matrix in noisy case.}\label{fig:2}
    \end{minipage}
\end{figure*}

In this subsection, we conduct simulations to validate the developed Theorems.
We generate the preamble signature matrix according to the complex Gaussian distribution unless otherwise stated, i.e., $ \bm{S} \sim \mathcal{C}\mathcal{N}(\bm{0}, \bm{I})$.
We set the length of the signature sequence, the total number of devices, and the number of antennas at the BS, i.e., $L,N$ and $M$, to $100$, $200$, and $30$, respectively. 
%Then, by fixing signature matrix $\bm{S}$, we compare our proposed LISTA-GS with other two baseline algorithms in both noiseless and noisy cases.

\textbf{Validation of Theorem \ref{thm:main1}}.
Theorem 1 endorses the empirical results shown in Figs. \ref{fig:tha} and \ref{fig:thb}.
In Fig. \ref{fig:v1}, the value of $ \| \bm{W}_{2}^{k}-(\bm I-\bm{W}_{1}^{k} \bm{\tilde{S}}) \|_2 $ and $ \theta^k$ in LISTA-GS are reported.
We observe that as $k$ increases, the value of $ \| \bm{W}_{2}^{k}-(\bm I-\bm{W}_{1}^{k} \bm{\tilde{S}}) \|_2 $ and $ \theta^k$ approach to $0$. 
The simulations clearly validate Theorem 1: $ \bm{W}_{2}^{k}-(\bm I-\bm{W}_{1}^{k} \bm{\tilde{S}}) \rightarrow \bm{0}$ and $\theta^k \rightarrow 0 $, as $ k \rightarrow \infty$.

\textbf{Validations of Theorems \ref{thm:convergence1} and \ref{thm:convergence2}}.
We examine the convergence rates of the proposed unrolled networks.
Fig. \ref{fig:1} shows the NMSE of the proposed unrolled networks and other methods over iterations in a noisy scenario when SNR $=15~$dB. 
%For the baseline ISTA-GS method, there exists an inherent trade-off between the convergence rate and the NMSE. 
%In particular, a smaller value of $\lambda$ results in a more accurate solution but leads to a lower convergence rate, and vice versa. 
In Fig. \ref{fig:1}, we show that the proposed networks (almost) converge linearly, which validate Theorem 2 and 3.
Furthermore, the LISTA-GS converges fastest, but with the most number of trainable parameters.
Besides, we observe that the proposed unrolled networks outperform other baseline algorithms in terms of NMSE.
\begin{table}[!t]
        \centering
        \footnotesize
        \caption{Training time and estimation performance of the proposed unrolled networks.}
        \begin{tabular}{c|c|c|c}
                \toprule
                Networks&LISTA-GS&LISTA-GSCP&ALISTA-GS\\
                \midrule  
                Training time (min)&160.13&105.31&87.09\\
                \midrule
                NMSE (dB)&-24.08&-23.38&-23.30\\
                \bottomrule
        \end{tabular}
        \label{ta:3}
\end{table}
\begin{table*}[!t]
		\centering
        \footnotesize
        \caption{Average running time of various methods of $12$-layer iterations per sample.}
        \begin{tabular}{c|c|c|c|c|c|c}
                \toprule
                LISTA-GS&LISTA-GSCP&ALISTA-GS&ISTA-GS&Nesterov's method&AMP-MMV&AMP MMSE-denoiser\\
                \midrule  
                0.0078s&0.0076s&0.0076s&0.0079s&0.0087s&0.2089s&0.0731s\\
                \bottomrule
        \end{tabular}
        \label{ta:2}
\end{table*}
\begin{figure*}[t]
        \centering
        \begin{minipage}{.46\textwidth}
                \centering
                \includegraphics[width=1.0\columnwidth]{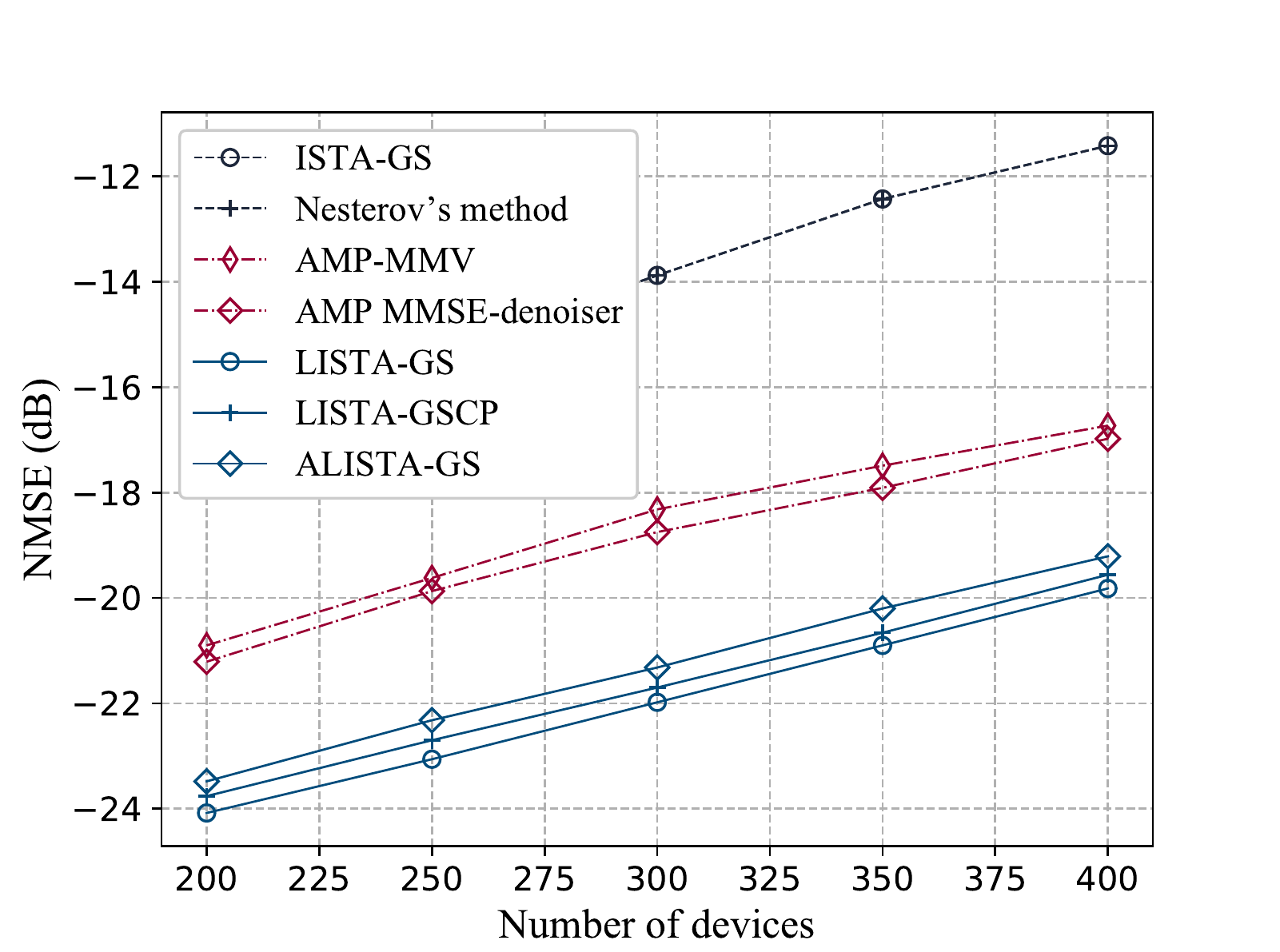}
                \caption{NMSE versus number of devices.}\label{fig:reply1}
        \end{minipage}
        \begin{minipage}{.46\textwidth}
                \centering
                \includegraphics[width=1.0\columnwidth]{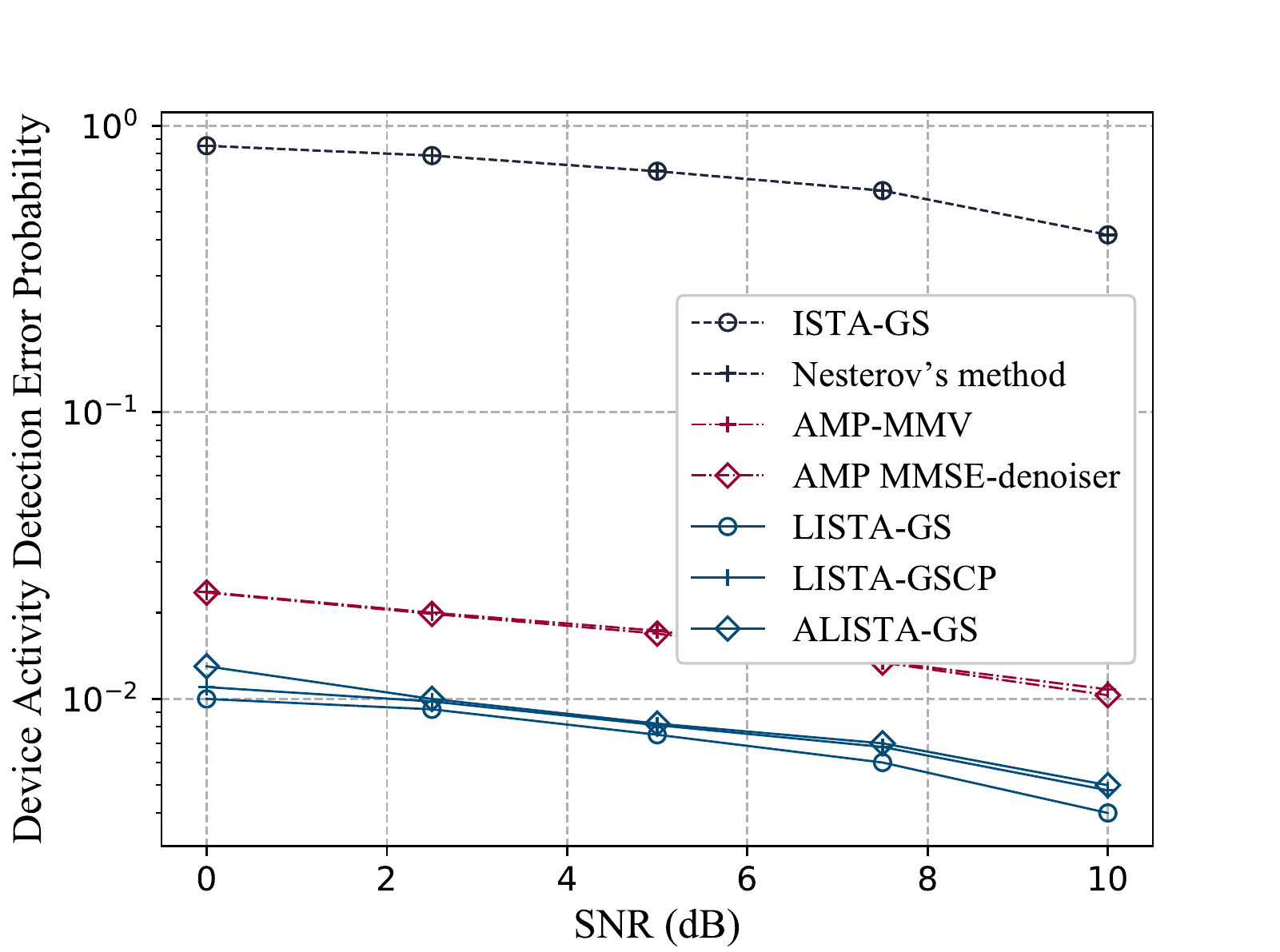}
                \caption{Device activity detection error probability versus SNR.}\label{fig:reply2}
        \end{minipage}
        
        \begin{minipage}{.46\textwidth}
                \centering
                \includegraphics[width=1.0\columnwidth]{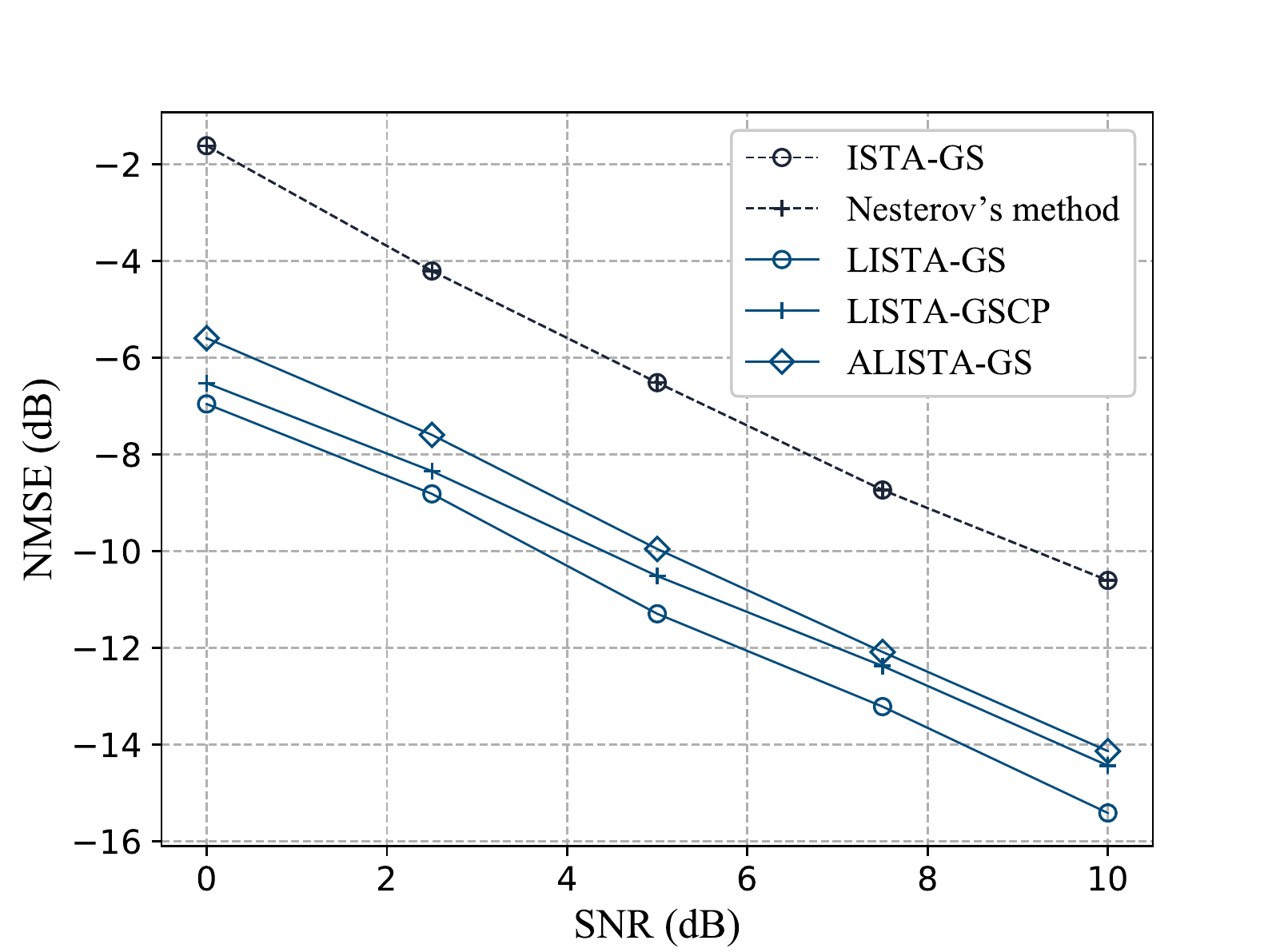}
                \caption{NMSE versus SNR when preamble signature $\bm{S}$ is binary matrix in noisy case. AMP-MMV and AMP MMSE-denoiser fail in this case. }\label{fig:binary}
        \end{minipage}
        \begin{minipage}{.46\textwidth}
                \centering
                \includegraphics[width=1.0\columnwidth]{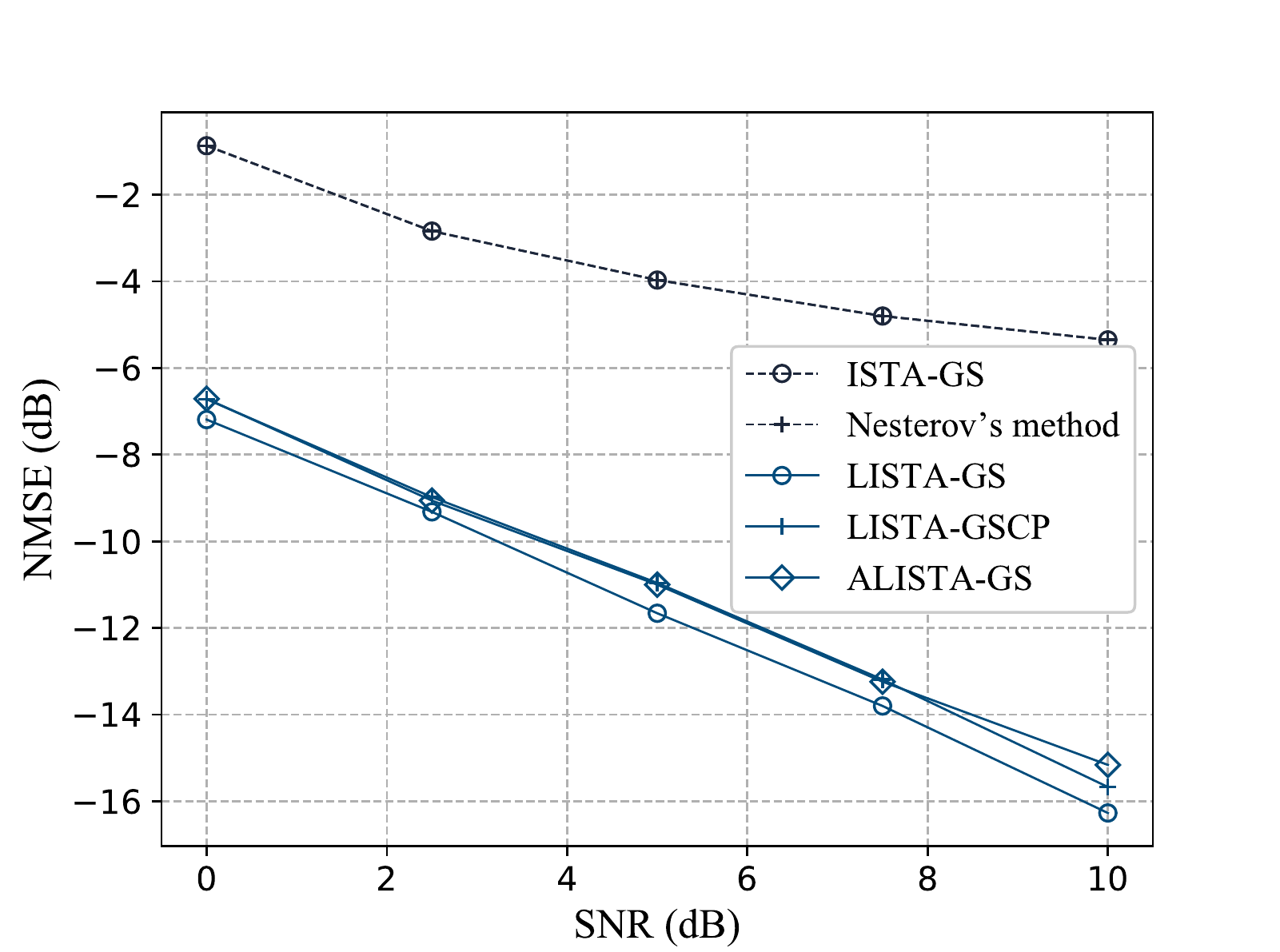}
                \caption{NMSE versus SNR when preamble signature $\bm{S}$ is Zadoff-Chu sequence matrix in noisy case. AMP-MMV and AMP MMSE-denoiser fail in this case.}\label{fig:zc}
        \end{minipage}
\end{figure*}
\textbf{Training Time Comparison.} We compare the training time of the proposed unrolled networks.
We train all the networks of $12$ layers with $64$ training samples in noisy scenario when SNR = $15$ dB on GeForce GTX 1080.
Table \ref{ta:3} shows that the ALISTA-GS is not only faster to train but performs as well as LISTA-GS and LISTA-GSCP.
With less trainable parameters, the training process turns to be faster.

\textbf{Computation Time Comparison.}
We compare the running time of test stage for different methods.
Once the proposed unrolled networks are trained, we can use them to recover many new sample of different channels. 
We run all the methods of $12$ iterations on Intel(R) Core(TM) i7-8650U CPU @ 2.11 GHz and average over $100$ test samples, which is shown in Table \ref{ta:2}.
The running time of the proposed unrolled networks is very close to the ISTA-GS and Nesterov's method, which validate the time complexity analysis in Section III-C.
In addition, the proposed unrolled networks run faster than AMP-based algorithms.
This benefits from the low computational complexity per iteration of ISTA-GS.

\textbf{Convergence Performance with Ill-Conditioned Preamble Signature.}
We consider the simulations when preamble signature matrix $\bm{S}$ is an ill-conditioned matrix, which is defined as a matrix with a large condition number $\kappa$. We consider the fixed signature matrix $\bm{S}$ with condition numbers $ \kappa = 2$ and $ 15$ in this setting.
To obtain the signature matrix $\bm{S}$ of condition number $ \kappa = 2$ and $ 15$, we first sample a matrix $\bm{A} \in \C^{L \times N}$, i.e., $\bm{A} \sim \mathcal{C}\mathcal{N}(\bm{0}, \bm{I})$.
Then, we decompose $\bm{A}=\bm{U\varSigma V^{\natural}}$ by singular value decomposition and replace $ \bm{\varSigma}$ by a new $ \bm{\varSigma^{\natural}}$ that satisfy the above conditions.

Fig. \ref{fig:k2} and Fig. \ref{fig:k50} show the simulation results of NMSE when signature matrix $\bm{S}$ of the condition number is $ \kappa = 2$ and $ \kappa = 15$ with SNR $=15$ dB, respectively.
The baseline AMP-MMV and AMP MMSE-denoiser fail when $ \kappa = 15$ in simulations.
Comparing with the case $ \kappa = 2$ and $ \kappa = 15$, the proposed unrolled networks remain stable but the outputs of ISTA-GS and Nesterov's method diverge when $ \kappa = 15$.
This results show that the proposed unrolled networks are robust to ill-conditioned preamble signature matrix.

%This is because the proposed method treats $\lambda$ as a weight in the training process, and hence obtains a good solution that balances the trade-off. 

\subsection{Performance of Proposed Unrolled Neural Networks}
In this subsection, we show the performance in terms of NMSE, and compare the proposed methods with other classic CS-based methods.
By varying the SNR from $0$ to $10$dB, all the algorithms under consideration reach a stable solution in the simulations, and we set $K=12$.
The smaller NMSE value means the better recovery performance of JADCE.
In these cases, we set $L,N$ and $M$ to $90$, $300$, and $100$, which is suitable for grant-free massive access.

\textbf{Gaussian Matrix as Preamble Signature}.
We conduct simulations for the case in which preamble signature matrix $\bm{S}$ is a complex Gaussian matrix, i.e., $ \bm{S} \sim \mathcal{C}\mathcal{N}(\bm{0}, \bm{I})$.
Fig. \ref{fig:2} shows the impact of SNR on the NMSE of the proposed unrolled networks and baseline algorithms. 
First, the performance of NMSE decreases as SNR grows, which implies that the JADCE performance becomes better as SNR increases.
Second, the proposed network structures achieve a much lower NMSE than other CS-based methods for different values of SNR. 
In addition, the LISTA-GS outperforms the LISTA-GSCP and ALISTA-GS.
This indicates that the more trainable parameters, the better NMSE performance.

Fig. \ref{fig:reply1} demonstrates the NMSE performance versus number of devices of the proposed unrolled networks and baseline algorithms with $L = 90$ and SNR = $15$ dB.
As a result, the NMSE increases as the number of devices increase for all the methods, while the proposed unrolled networks outperform other baseline algorithms.
Fig. \ref{fig:reply2} shows that the proposed unrolled networks have the less device activity detection error probability compared with other baseline methods, which also verify the superior of the proposed unrolled networks.

\textbf{Binary Matrix as Preamble Signature}.
We conduct simulations for the case in which the preamble sequence matrix $\bm{S}$ is binary sequence matrix, i.e., $\bm{S} \in \{ \pm 1 \}^{L \times N}$.
Each entry of $\bm{S}$ is selected uniformly at random on $\{ \pm 1 \}$, which is closely related to Code Division Multiple Access (CDMA) communication systems \cite{kabashima2003cdma}.
Fig. \ref{fig:binary} shows the NMSE performance of the proposed unrolled networks and baseline methods.
As a result, the proposed networks remain stable but AMP-MMV and AMP MMSE-denoiser fail to solve the JADCE problem when preamble sequence matrix $\bm{S}$ is binary.
The proposed unrolled networks outperform baseline algorithms more than $ 3$dB on NMSE for different values of SNR.
Obviously, the proposed unrolled networks achieve significantly better NMSE performance compared with CS-based methods.
This result demonstrates the robustness of our proposed methods for the non-Gaussian matrices case.

\textbf{Zadoff-Chu Sequences as Preamble Signature.}
In this case, we carry out simulations when preamble signature $\bm{S}$ are composed of Zadoff-Chu sequences\cite{ding2019comparison}, which are widely used in practical situations, i.e., 4G LTE systems.
Fig. \ref{fig:zc} shows the NMSE performance of the proposed unrolled networks and other baseline CS-based methods.
The proposed networks remain stable but AMP-based algorithms fail to solve the JADCE when $\bm{S}$ is Zadoff-Chu sequence matrix.
From this figure, we observe that the proposed unrolled networks achieve a much better NMSE performance compared with the baseline methods.
Moreover, as the SNR increasing, the proposed approaches perform better than baseline methods.
Among the proposed structures, the LISTA-GS achieves the best performance and LISTA-GSCP shares a similar performance to that of ALISTA-GS.
%Obviously, the proposed unrolled networks achieves significantly better NMSE performance comparing with baseline methods.
This result also shows the robustness of our proposed methods for practical situations.

In summary, the simulations demonstrate
the effectiveness of the proposed unrolled networks in the
following aspects:
\begin{itemize}
        \item The proposed unrolled networks converge faster than the robust algorithms such as ISTA-GS and Nesterov's method, yielding much lower computational complexity.
        \item Simulations of different preambles show that the proposed unrolled networks are more robust compared with the computationally efficient algorithms such as AMP-based algorithms.
        \item The proposed unrolled networks achieve better performance for JADCE comparing with the baseline CS-based algorithms.
\end{itemize}

\section{Conclusions}\label{sec:conclu}
In this paper, we proposed a novel unrolled deep neural network framework enjoying linear convergence rate, low computational complexity and high robustness to solve the JADCE problem in grant-free massive access for IoT networks.
%To address the high computational cost and lack of algorithm robustness of the exiting CS-based algorithm simultaneously, we developed a novel unrolled deep neural network framework with three structures to solve JADCE problem in grant-free massive access for IoT networks.
We introduced the first unrolled network structure by mapping the iterative algorithm ISTA-GS as an unrolled RNN, thereby improving the convergence rate by end-to-end training.
To make training procedure efficiently and tackle the interpretability issue of deep learning, we proposed two simplified unrolled network structures with less trainable parameters and proved the linear convergence rate of these methods.
Extensive simulations were
conducted to verify the effectiveness of the proposed unrolled networks in terms of the convergence rate, robustness and estimation accuracy for the JADCE problem.
%Simulation results show that the proposed neural network structures achieve
%significantly better performance and more robustness comparing with the classic CS-based methods for JADCE in massive access.

%%%%%%%%%%%%%%%%%%%%%%%%%%%%%%%%%%
%This is APPENDIX part!
%%%%%%%%%%%%%%%%%%%%%%%%%%%%%%%%%%

\begin{appendices}

\section{Proof of theorem \ref{thm:main1}}\label{pr:th1}
\begin{proof}
        We prove the threshold value $\theta^k$ converges to zero firstly, and then we will show that the weights  $\{\bm W_{1}^{k}, \bm W_{2}^{k} \}_{k=0}^{\infty}$ in LISTA-GS
        have weight coupling property.

        (1) We verify that $ \theta^k  \to 0$ as $k \to \infty$ in \eqref{thm1:result1}.
        We define a subset of $\mathcal{X}(\beta,s,\sigma)$ for a given $0<\tilde{\beta} \leq \beta$ as          
        \begin{align*}
        \tilde{\mathcal{X}}(\beta, \tilde{\beta}, s, \sigma)   
        :=\big\{(\bm{\tilde X}^{\natural},  \bm{\tilde{Z}}) \in \mathcal{X}(\beta,s,\sigma) |~ | \supp(\psi(\bm{\tilde X}^{\natural})) | \leq s, \\
        \tilde{\beta} \leq  \big\| \bm{\tilde X}^{\natural}[i,:]   \big \|_2 \leq \beta,  
        \forall i \in \supp(\psi(\bm{\tilde X}^{\natural})) \big\}. 
        \end{align*}
                
        Clearly,  $\tilde{\mathcal{X}}(\beta, \tilde{\beta}, s, 0) \subset \mathcal{X}( \beta, s, 0)$.
        Since $\bm{\tilde X}^{k} \rightarrow \bm{\tilde X}^{\natural}$ is uniform for all $(\bm{\tilde X}^{\natural}, \bm{0}) \in \mathcal{X}(\beta, s, 0),$ we have $(\bm{\tilde X}^{\natural}, \bm{0}) \in \tilde{\mathcal{X}}(\beta, \tilde{\beta}, s, 0)$, where $\tilde{\beta} \leq \beta$. 
        Then, there exists $K_{1} \in \mathbb{N}$ for all $(\bm{\tilde X}^{\natural}, \bm{0}) \in \tilde{\mathcal{X}}(\beta, \beta /10, s, 0)$ such that
        if $k \geq K_{1}$, 
        $$\Big| \big\|  \bm{\tilde X}^{k}[i,:] \big\|_2 - \big\| \bm{\tilde X}^{\natural}[i,:] \big\|_{2} \Big| < \frac{\beta}{10} \quad \forall i\in [2N].$$
        Then we have 
        $
        \operatorname{sign}( \psi(\bm{\tilde X}^{k}))=\operatorname{sign}(\psi(\bm{\tilde X}^{\natural})), \forall k \geq K_{1}.
        $
        Recall that the recurrence relation
        $
        \bm{\tilde X}^{k+1} = \eta_{\theta^k}(\bm{W}_1^k\bm{\tilde{Y}}+\bm{W}_2^k \bm {\tilde X}^k)
        $
        and let $\mathcal{I} =\supp(\psi(\bm{\tilde{X}}^{\natural}))$. From the uniform convergence of $\bm{\tilde X}^{k}$, it follows that for any $k \geq K_{1} \text { and }(\bm{\tilde X}^{\natural}, \bm{0}) \in \tilde{\mathcal{X}}(\beta, \beta/10, s, 0)$. We have
        \begin{align}
        \bm{\tilde X}^{k+1}[\mathcal{I},:]
        %&=\eta_{\theta^{k}}(\bm{W}_{2}^{k}[\mathcal{I}, \mathcal{I}] \bm{\tilde X}^{k}[\mathcal{I},:]+\bm{W}_{1}^{k}[\mathcal{I}, :] \bm{\tilde Y})     \\ \nonumber
        %&=\eta_{\theta^{k}} (  \bm{\tilde Q}^k )       \\ 
        = \bm{\tilde Q}^k - \theta^k
        {\diagg}
        (\psi (\bm{\tilde Q}^k))^{-1}    \bm{\tilde Q}^k,  \label{eq:iter1}     
        \end{align}
        where
        $ \bm{\tilde Q}^k = \bm{W}_{2}^{k}[\mathcal{I}, \mathcal{I}] \bm{\tilde X}^{k}[\mathcal{I},:]+\bm{W}_{1}^{k}[\mathcal{I}, :] \bm{\tilde Y}$.

        Also, the uniform convergence of $\bm{\tilde X}^{k}$ implies that for any $\varepsilon>0$ and $(\bm{\tilde X}^{\natural}, \bm{0}) \in \tilde{\mathcal{X}}(\beta, \beta/10, s, 0),$
        there exists $K_{2} \in \mathbb{N}$ such that if $k\geq K_2$, then $\| \bm{\tilde X}^{k}[\mathcal{I},:]-\bm{\tilde X}^{\natural}[\mathcal{I},:] \|_F < \varepsilon$. 
        %and $\| \bm{\tilde X}^{k+1}[S,:] - \bm{\tilde X}^{\natural}[S,:] \|_F \leq \varepsilon$.
        Denote $\bm{\tilde{E}}^k = \bm{\tilde X}^{k}[\mathcal{I},:]-\bm{\tilde X}^{\natural}[\mathcal{I},:] \in \mathbb{R}^{|\mathcal{I}|\times M}$ for each $k \in \mathbb{N} $.
        %       and  $\bm{\tilde{E}}^{k+1} =  \bm{\tilde X}^{k+1}[S,:]-\bm{\tilde X}^{\natural}[S,:] \in  \mathbb{R}^S$.
        That is, $\| \bm{\tilde{E}}^k \|_F < \varepsilon$ for all $k\geq K_2$. %and  $\| \bm{\tilde{E}}_{k+1} \|_F \leq \varepsilon$.
                        
        Since the noise is assumed to be zero, i.e., $\bm{\tilde{Y}} = \bm{\tilde{S}} \bm{\tilde{X}}^{\natural}$. 
        Then from
        \eqref{eq:iter1} we have
        \begin{align}
        &\| \theta^k 
        {\diagg}
        (\psi (\bm{\tilde Q}^k))^{-1}    \bm{\tilde Q}^k[:,{j}] \|_2^2 
        %&= \big\| -\big(\bm{I} - \bm{W}_{2}^{k}[\mathcal{I}, \mathcal{I}]-  \bm{W}_{1}^{k}[\mathcal{I}, :] \bm{\tilde S}[:,\mathcal{I}] \big) \bm{\tilde X}^{\natural}[\mathcal{I},{j}] \\ \nonumber
        %& \quad +\bm{W}_{2}^{k}[\mathcal{I}, \mathcal{I}]\bm{\tilde{E}}^{k}[:,j]   -\bm{\tilde{E}}^{k+1}[:,j] \big\|_2^2 \\
        =\xi_j^k    + \notag \\
         &\big\| \big(\bm{I} - \bm{W}_{2}^{k}[\mathcal{I}, \mathcal{I}]-  \bm{W}_{1}^{k}[\mathcal{I}, :] \bm{\tilde S}[:,\mathcal{I}] \big) \bm{\tilde X}^{\natural}[\mathcal{I},{j}] \big\|_2^2, \label{eq:iter2}
        \end{align}
        where 
        $
        \xi_j^k 
        = \big\| -\big(\bm{I} - \bm{W}_{2}^{k}[\mathcal{I}, \mathcal{I}]-  \bm{W}_{1}^{k}[\mathcal{I}, :] \bm{\tilde S}[:,\mathcal{I}] \big) \bm{\tilde X}^{\natural}[\mathcal{I},{j}] +\bm{W}_{2}^{k}[\mathcal{I}, \mathcal{I}]\bm{\tilde{E}}^{k}[:,j]   -\bm{\tilde{E}}^{k+1}[:,j]  \big\|_2^2 
        - \big\| \big(\bm{I} - \bm{W}_{2}^{k}[\mathcal{I}, \mathcal{I}]-  \bm{W}_{1}^{k}[\mathcal{I}, :] \bm{\tilde S}[:,\mathcal{I}] \big) \bm{\tilde X}^{\natural}[\mathcal{I},{j}] \big\|_2^2 . \label{eq:xi_jk1}
        $       
                        
        We now show that $| \xi_j^k | $ is sufficiently small for $k\geq K_2$.
        Denote extra notations 
        \begin{align*}
        &\bm{p}_j^k=(\bm{I} - \bm{W}_{2}^{k}[\mathcal{I}, \mathcal{I}]-  \bm{W}_{1}^{k}[\mathcal{I}, :] \bm{\tilde S}[:,\mathcal{I}] \big) \bm{\tilde X}^{\natural}[\mathcal{I},{j}], \quad \\
        &\bm{q}_j^k=\bm{W}_{2}^{k}[\mathcal{I}, \mathcal{I}]\bm{\tilde{E}}^{k}[:,j]   -\bm{\tilde{E}}^{k+1}[:,j].
        \end{align*}
        The value $\xi_j^k$ can be simply rewritten as 
  $\| \bm{q}_j^k \|_2^2 -2\langle \bm{p}_j^k,\bm{q}_j^k \rangle, 
        $
        where $\langle \cdot, \cdot \rangle$ is the inner product in $\mathbb{R}^{|\mathcal{I}|}$.
        Then the Cauchy-Schwartz inequality implies $
                |\xi_j^k |  \leq   (2 \| \bm{p}_j^k \|_2 + 1) \| \bm{q}_j^k \|_2.$
        
        %We will show that (i) $\| \bm{p}_j^k \|_2$ is bounded for all $k$.(ii) $\| \bm{q}_j^k \|_2$ is sufficiently small for $k\geq K_2$, implying that the value $ |\xi_j^k |$ is sufficiently small for $k\geq K_2$.
        From the fact that  $\| \bm{A} \|_{2} \leq  \| \bm{A} \|_{F}$ for any matrix $\bm{A}$ and the triangle inequality it follows that if $k\geq K_2$, then
        \begin{align}
        \| \bm{q}_j^k \|_2
         %= \| \bm{W}_{2}^{k}[\mathcal{I}, \mathcal{I}]\bm{\tilde{E}}^{k}[:,j]   -\bm{\tilde{E}}^{k+1}[:,j]   \|_2  
         &\leq \| \bm{W}_{2}^{k}[\mathcal{I}, \mathcal{I}] \|_{2} \| \bm{\tilde{E}}^{k}[:,j] \|_{2}  + \| \bm{\tilde{E}}^{k+1}[:,j]  \|_2 \notag \\
        %& \leq \| \bm{W}_{2}^{k}[\mathcal{I}, \mathcal{I}] \|_{F} \| \bm{\tilde{E}}^{k}[:,j] \|_{2}  + \| \bm{\tilde{E}}^{k+1}[:,j]  \|_2 \nonumber \\
        &\leq (C_{W_2} +1)\varepsilon. \label{eq:q-bounded1}
        \end{align}
        Note that               
        \begin{align}
        \| \bm{p}_j^k \|_2^2
         %= \big\|(\bm{I} - \bm{W}_{2}^{k}[\mathcal{I}, \mathcal{I}]-  \bm{W}_{1}^{k}[\mathcal{I}, :] \bm{\tilde S}[:,\mathcal{I}] \big) \bm{\tilde X}^{\natural}[\mathcal{I},{j}]  \big\|_2^2  \nonumber \\
         \leq \big\| \bm{I} - \bm{W}_{2}^{k}[\mathcal{I}, \mathcal{I}]-  \bm{W}_{1}^{k}[\mathcal{I}, :] \bm{\tilde S}[:,\mathcal{I}]  \big\|_2^2    \big \| \bm{\tilde X}^{\natural}[\mathcal{I},{j}]  \big\|_2^2.  
        \end{align}
        Then it follows that
        \begin{align}
        \sum_{j=1}^M \| \bm{p}_j^k \|_2^2  
        %& \leq \big\| \bm{I} - \bm{W}_{2}^{k}[\mathcal{I}, \mathcal{I}]-  \bm{W}_{1}^{k}[\mathcal{I}, :] \bm{\tilde S}[:,\mathcal{I}]  \big\|_2^2    \big \| \bm{\tilde X}^{\natural}[\mathcal{I},{:}]  \big\|_F^2  \nonumber \\
        %& \leq \Big( \big\| \bm{I} \big\|_2 +  \big\| \bm{W}_{2}^{k}[\mathcal{I}, \mathcal{I}] \big\|_2 \nonumber \\
        %& \qquad + \big\|  \bm{W}_{1}^{k}[\mathcal{I}, :] \big\|_2  \big\| \bm{\tilde S}[:,\mathcal{I}]  \big\|_2 \Big)^2   \big \| \bm{\tilde X}^{\natural}[\mathcal{I},{:}]  \big\|_F^2 \nonumber \\
         &\leq \Big( 1 +  \big\| \bm{W}_{2}^{k} \big\|_F
        + \big\|  \bm{W}_{1}^{k}\big\|_F  \big\| \bm{\tilde S} \big\|_F \Big)^2   \big \| \bm{\tilde X}^{\natural} \big\|_F^2  \notag\\
        & \leq \Big( 1 + C_{W_2}
        +C_{W_1} \big\| \bm{\tilde S} \big\|_F \Big)^2   \big \| \bm{\tilde X}^{\natural} \big\|_F^2. 
        \end{align}
        By Cauchy-Schwarz inequality, it holds that 
        \begin{align}
        \sum_{j=1}^M \| \bm{p}_j^k \|_2
         &\leq  \sqrt{M} \sqrt{\sum_{j=1}^M \| \bm{p}_j^k \|_2^2} \notag \\
         &\leq  \sqrt{M} \Big( 1 + C_{W_2}
        +C_{W_1} \big\| \bm{\tilde S} \big\|_F \Big) \big \| \bm{\tilde X}^{\natural} \big\|_F.  \label{eq:p-bounde1}
        \end{align}
        Thus, by \eqref{eq:q-bounded1} and \eqref{eq:p-bounde1}
        \begin{align}
        \sum_{j=1}^M \big|\xi_j^k \big| 
         &\leq (2 \| \bm{p}_j^k \|_2 + 1) \| \bm{q}_j^k \|_2 \notag \\
         &\leq (C_{W_2}+1)  \varepsilon \sum_{j=1}^M (2 \| \bm{p}_j^k \|_2 +1 ) \notag \\
        %& \leq (C_{W_2}+1)  \varepsilon \Big( 2\sum_{j=1}^M  \| \bm{p}_j^k \|_2 +M \Big) \nonumber \\
         &\leq (C_{W_2}+1)   ( 2C +M )\varepsilon, \label{eq:xi-bounded2}
        \end{align}
        where $C=\sqrt{M} \Big( 1 + C_{W_2}
        +C_{W_1} \big\| \bm{\tilde S} \big\|_F \Big) \big \| \bm{\tilde X}^{\natural} \big\|_F$.      
        Then, by \eqref{eq:iter2} we have
        \begin{align*}
        \big\| \big(\bm{I} - \bm{W}_{2}^{k}[\mathcal{I}, \mathcal{I}]-  \bm{W}_{1}^{k}[\mathcal{I}, :] \bm{\tilde S}[:,\mathcal{I}] \big) \bm{\tilde X}^{\natural}[\mathcal{I},:] \big\|_F^2 
        %&= \sum_{j=1}^M \big\| \big(\bm{I} - \bm{W}_{2}^{k}[\mathcal{I}, \mathcal{I}]-  \bm{W}_{1}^{k}[\mathcal{I}, :] \bm{\tilde S}[:,\mathcal{I}] \big) \bm{\tilde X}^{\natural}[\mathcal{I},{j}] \big\|_2^2 \\
        %&=  -\sum_{j=1}^M \xi_j^k + \sum_{j=1}^M  \| \theta^k {\diagg}(\psi (\bm{\tilde Q}^k))^{-1}  \bm{\tilde Q}^k[:,{j}] \|_2^2 \\
        %&=  -\sum_{j=1}^M \xi_j^k + \big|\theta^k \big|^2 \|  
        %{\diagg}
        %(\psi (\bm{\tilde Q}^k))^{-1}   \bm{\tilde Q}^k \|_F^2 \\
       \\ =  -\sum_{j=1}^M \xi_j^k + \big|\theta^k \big|^2 \sum_{i=1}^{| \mathcal{I} |}  \Big\|
        \frac{1}{\| \bm{\tilde Q}^k[i,:] \|_2}
        \bm{\tilde Q}^k[i,:] \Big\|_2^2 
        \end{align*}    
        Therefore, by \eqref{de:mixnorm} we get
        \begin{align}
        &\big\| \big(\bm{I} - \bm{W}_{2}^{k}[\mathcal{I}, \mathcal{I}]-  \bm{W}_{1}^{k}[\mathcal{I}, :] \bm{\tilde S}[:,\mathcal{I}] \big) \bm{\tilde X}^{\natural}[\mathcal{I},:] \big\|_F^2 \notag \\
        &=  -\sum_{j=1}^M \xi_j^k +  | \theta^k |^2  | \mathcal{I} |. \label{eq:2Xpsi-xi1-xi1}
        \end{align}
        For any $(\bm{\tilde X}^{\natural}, \bm{0}) \in \tilde{\mathcal{X}}(\beta / 2, \beta /10, s, 0)$, it is true that $(2 \bm{\tilde X}^{\natural}, \bm{0}) \in \tilde{\mathcal{X}}(\beta, \beta/ 10, s, 0)$ holds. Thus, the above argument holds for all $(2\bm{\tilde X}^{\natural},\bm{0})$ if $(\bm{\tilde X}^{\natural}, \bm{0}) \in \tilde{\mathcal{X}}(\beta / 2, \beta / 10, s, 0) .$         
        Substituting $\bm{\tilde X}^{\natural}$ with $2\bm{\tilde{X}}^{\natural}$ in the previous equation, we have
        \begin{align}
        &4\big\| \big(\bm{I} - \bm{W}_{2}^{k}[\mathcal{I}, \mathcal{I}]-  \bm{W}_{1}^{k}[\mathcal{I}, :] \bm{\tilde S}[:,\mathcal{I}] \big) \bm{\tilde X}^{\natural}[\mathcal{I},:] \big\|_F^2 \notag \\
         &=  -\sum_{j=1}^M \hat{\xi}_j^k +  | \theta^k |^2  | \mathcal{I} |. \label{eq:2Xpsi-xi1-xi2}
        \end{align}
        From \eqref{eq:2Xpsi-xi1-xi1} and \eqref{eq:2Xpsi-xi1-xi2} one can get
        $
        3 | \theta^k|^2  | \mathcal{I} | = 4\sum_{j=1}^M {\xi}_j^k - \sum_{j=1}^M \hat{\xi}_j^k.
        $
        Then, taking the absolute value on the both sides,
        %\begin{eqnarray}
        %|\theta^{k}|^2 \leq \frac{1}{3  | \mathcal{I} | }  \sum_{j=1}^M (4 |{\xi}_j^k | +|\hat{\xi}_j^k| )  \notag
        %\end{eqnarray}
        then by \eqref{eq:xi-bounded2}, if  $\forall k \geq \max\{K_{1}, K_{2}\}$,               
        \begin{eqnarray}\label{eq:theta-k-bounded1}
        |\theta^{k}|^2 
                \leq \frac{5(C_{W_2}+1)   ( 2C +M )}{3  | \mathcal{I} | } \varepsilon.
        \end{eqnarray}
                
        Moreover, as $\theta^{k} $ is MSTO parameter, $\theta^{k} \ge 0 $.
        Therefore, $\theta^{k} \rightarrow 0$ as $k \rightarrow \infty$.
                
        %\bigskip       
        (2) We prove that $\bm I-\bm W_{2}^{k}-\bm W_{1}^{k} \bm{\tilde{S}} \rightarrow \bm{0}$ as $k \rightarrow \infty$.        
        LISTA-GS model \eqref{eq:ista} gives%and $\bm{\tilde{Y}}=\bm{\tilde{S}} \bm{\tilde{X}}^{\natural}$ gives
        \begin{align*}
        &\bm{\tilde X}^{k+1}[\mathcal{I},:]
        %&= \eta_{\theta^k} \big( \bm{W}_1^k [\mathcal{I},:] \bm{\tilde{Y}} + \bm{W}_2^k[\mathcal{I},:]  \bm{\tilde{X}}^k  \big) \\
        = \eta_{\theta^k} \big( \bm{W}_1 [\mathcal{I},:] ^k \bm{\tilde{S}})  \bm{\tilde{X}}^{\natural}  + \bm{W}_2^k [\mathcal{I},:]  \bm{\tilde{X}}^k  \big) \\
         &\in   \bm{W}_1^k [\mathcal{I},:]  \bm{\tilde{S}}  \bm{\tilde{X}}^{\natural}  + \bm{W}_2^k [\mathcal{I},:]  \bm{\tilde{X}}^k -\theta^k \partial \ell_{2,1}(\bm{\tilde X}^{k+1}[\mathcal{I},:]) , 
        \end{align*}
        where $\partial \ell_{2,1}(\bm{\tilde{X}})$ is the sub-gradient of $\| \bm{\tilde{X}} \|_{2,1}$. It is a set defined
        for each row as follows:
        \begin{eqnarray}\label{de:mixnorm}
        \partial \ell_{2,1}(\bm{X})
        =
        \begin{bmatrix}
        \partial \| \bm{X}[1,:] \|_2 ,
        \partial \| \bm{X}[2,:] \|_2 ,
        \cdots ,
        \partial \| \bm{X}[2N,:] \|_2 
        \end{bmatrix}^T,
        \end{eqnarray}
        where
        \begin{align}
        \partial \| \bm{X}[n,:] \|_2
        =
        \begin{cases}
        \big\{ \frac{\bm{X}[n,:]}{\| \bm{X}[n,:] \|_2} \big \} \quad \text{if }\bm{X}[n,:] \neq \bm{0}, \\
        \{ \bm{h} \in \mathbb{R}^M |~ \| \bm{h} \|_2 \leq 1 \}  \quad \text{if }\bm{X}[n,:] = \bm{0}.
        \end{cases}
        \label{l2norm}
        \end{align}
        If $k \geq \max \{K_1,K_2 \}$ then from \eqref{eq:theta-k-bounded1} and     \eqref{eq:xi-bounded2}
        \begin{align*}
        &\Big\| \big(\bm{I} - \bm{W}_{2}^{k}[\mathcal{I}, \mathcal{I}]-  \bm{W}_{1}^{k}[\mathcal{I}, :] \bm{\tilde S}[:,\mathcal{I}] \big) \bm{\tilde X}^{\natural}[\mathcal{I},j]  \Big\|_2 \\
        %&=  -\xi_j^k +  \| \theta^k {\diagg}(\psi (\bm{\tilde Q}^k))^{-1}       \bm{\tilde Q}^k[:,{j}] \|_2^2 \\
        %& \leq  |\xi_j^k |+ | \theta^k|^2 \|  {\diagg}(\psi (\bm{\tilde Q}^k))^{-1}      \bm{\tilde Q}^k[:,{j}] \|_2^2 \\
        %& \leq  |\xi_j^k |+ | \theta^k|^2   \\
         &\leq  \Big( (C_{W_2}+1)   ( 2C +M )+\frac{5(C_{W_2}+1)   ( 2C +M )}{3  | \mathcal{I} | } \Big) \varepsilon.
        \end{align*}
        %The first equality holds from \eqref{eq:iter2} and the last inequality holds . 
                
        From the definition of operator norm we have that               
        \begin{align*}
        &\sigma_{\max} 
        \Big(\bm{I} - \bm{W}_{2}^{k}[\mathcal{I}, \mathcal{I}]-  \bm{W}_{1}^{k}[\mathcal{I}, :] \bm{\tilde S}[:,\mathcal{I}] \Big)  \\
        %&=\underset{\underset{\supp(\bm{\tilde X}^{\natural}[\mathcal{I},j] ) =s}{\|\bm{\tilde X}^{\natural}[\mathcal{I},j] \|_2=\beta}}{\sup}      \frac{ \big\| (\bm{I} - \bm{W}_{2}^{k}[\mathcal{I}, \mathcal{I}]-  \bm{W}_{1}^{k}[\mathcal{I}, :] \bm{\tilde S}[\mathcal{I},\mathcal{I}] ) \bm{\tilde X}^{\natural}[\mathcal{I},j]  \big\|_2}{\beta}   \\
        %&\leq \underset{ (\bm{\tilde X}^{\natural},0) \in \mathcal{X} (\beta, s, 0)  }{\sup}   \frac{ \big\| (\bm{I} - \bm{W}_{2}^{k}[\mathcal{I}, \mathcal{I}]-  \bm{W}_{1}^{k}[\mathcal{I}, :] \bm{\tilde S}[\mathcal{I},\mathcal{I}] ) \bm{\tilde X}^{\natural}[\mathcal{I},j] \big\|_2}{\beta} 
        &\leq \frac{1}{\beta} \Big( (C_{W_2}+1)   ( 2C +M )+\frac{5(C_{W_2}+1)   ( 2C +M )}{3  | \mathcal{I} | } \Big) \varepsilon.
        \end{align*}
        Thus, $s\geq 2$, $\bm{I} - \bm{W}_{2}^{k}[\mathcal{I}, \mathcal{I}]-  \bm{W}_{1}^{k}[\mathcal{I}, :] \bm{\tilde S}[:,\mathcal{I}] \rightarrow \bm{0}$ uniformly for all $\mathcal{I}$ with $2\leq |\mathcal{I}| \leq s$. Therefore,
                $\bm{I} - \bm{W}_{2}^{k}-  \bm{W}_{1}^{k}\bm{\tilde S} \rightarrow \bm{0}$ as $k \rightarrow \infty$.      
\end{proof}
%\medskip
\section{Proof of Lemma 1}\label{pr:le1}
We will show that for the no-false-positives property holds for the LISTA-GSCP.
\begin{proof}
        Let $(\bm{\tilde X}^{\natural},\bm{\tilde{Z}}) \in \mathcal{X}(\beta,s,\sigma)$ and $\mathcal{I} = \supp( \psi( \bm{\tilde X}^{\natural}) )$. We prove that if \eqref{g1} holds, then
        $
         \bm{\tilde{X}}^{k}[i,:] = \bm{0}, \forall i \notin \mathcal{I}, \forall k.
        $       
        
        (i) When $k = 0$, it is satisfied since $\bm{\tilde{X}}^{0}=\bm{0}$.
        
        (ii) Suppose that 
        $\bm{\tilde{X}}^k[i,:]  =\bm{0}$ for all $i \notin \mathcal{I}$.
        Then by \eqref{eq:LISTA-CP} it follows that 
        \begin{align*}
        \bm{\tilde X}^{k+1}[i,:] 
        %&= \eta_{\theta^k} \Big( \bm{\tilde{X}}^k[i,:]  + (\bm{W}^k[:,i])^T (\bm{\tilde{Y}}   -\bm{\tilde{S}}  \bm{\tilde{X}}^k)  \Big)\\
        %&= \eta_{\theta^k} \Big( \bm{\tilde{X}}^k[i,:]  - (\bm{W}^k[:,i])^T ( \bm{\tilde{S}}  \bm{\tilde{X}}^k - \bm{\tilde{S}}  \bm{\tilde{X}}^{\natural} - \bm{\tilde{Z}} )  \Big)\\
        = \eta_{\theta^k} \Big( \bm{\tilde{X}}^k[i,:]  - (\bm{W}^k[:,i])^T \bm{\tilde{S}} (  \bm{\tilde{X}}^k -   \bm{\tilde{X}}^{\natural}) \\
        +  (\bm{W}^k[:,i])^T \bm{\tilde{Z}}   \Big)
        \end{align*}
        for all $i \notin \mathcal{I}$.
        Then by the triangle inequality it follows that 
        \begin{align*}
        &\big\| - (\bm{W}^k[:,i])^T \bm{\tilde{S}} (  \bm{\tilde{X}}^k -   \bm{\tilde{X}}^{\natural}) +  (\bm{W}^k[:,i])^T \bm{\tilde{Z}} )  \big\|_2 \\
        %&=  \Big\| - \sum_{j\in \mathcal{I}} (\bm{W}^k[:,i])^T \bm{\tilde{S}}[:,j] (  \bm{\tilde{X}}^k[j,:] -   \bm{\tilde{X}}^{\natural}[j,:]) \\
        %& \qquad \qquad+  (\bm{W}^k[:,i])^T \bm{\tilde{Z}} )  \Big\|_2 \\
        &{\leq} \sum_{j\in \mathcal{I}}     \big| \bm{W}^k[:,i])^T \bm{\tilde{S}}[:,j] \big| \big\|  \bm{\tilde{X}}^k[j,:] -   \bm{\tilde{X}}^{\natural}[j,:] \big \|_2  +  \big\| (\bm{W}^k[:,i])^T \bm{\tilde{Z}} )  \big\|_2 \\
        %{\leq}  \sum_{j\in \mathcal{I}}    \tilde{\mu} \big| \big\|  \bm{\tilde{X}}^k[j,:] -   \bm{\tilde{X}}^{\natural}[j,:] \big \|_2  +  \big\| (\bm{W}^k[:,i])^T \bm{\tilde{Z}}   \big\|_2 
        %&{\leq}    \tilde{\mu} \big| \big\|  \bm{\tilde{X}}^k -   \bm{\tilde{X}}^{\natural} \big \|_{2,1} +  \big\| (\bm{W}^k[:,i])^T \bm{\tilde{Z}}  \big\|_2 \\
        %&{\leq}    \tilde{\mu} \big\|  \bm{\tilde{X}}^k -   \bm{\tilde{X}}^{\natural} \big \|_{2,1} + \big( \sum_{j=1}^M  \big| (\bm{W}^k[:,i])^T \bm{\tilde{Z}}[:,j] \big|^2 \big)^{\frac{1}{2}}  \\
        %&{\leq}     \tilde{\mu} \big| \big\|  \bm{\tilde{X}}^k -   \bm{\tilde{X}}^{\natural} \big \|_{2,1} + |C_W|  \| \bm{\tilde{Z}}\|_F \\
        &{\leq}    \tilde{\mu} \big| \big\|  \bm{\tilde{X}}^k -   \bm{\tilde{X}}^{\natural} \big \|_{2,1} +  C_W \| \bm{\tilde{Z}}\|_F .
        %&\leq  \theta^k,
        \end{align*}
        %where $\tilde{\mu}=\tilde{\mu} (\bm{\tilde{S}})$ is the generalized mutual coherence of $\bm{\tilde{S}}$ and $\| \cdot \|_{2,1}$ is the mixed $ \ell_1/\ell_2$-norm in \eqref{eq:Rmixed-norm1}.
        %From  it is easy to check $(a)$.
        
        Since $\theta^{k}=\tilde{\mu} \sup _{\bm{\tilde{X}}^{\natural}, \bm{\tilde{Z}} }\{\|\bm{\tilde{X}}^{k}-\bm{\tilde{X}}^{\natural}\|_{2,1}\}+C_{W} \sigma$ and $\bm{W}^{k} \in \mathcal{W}(\bm{\tilde{S}})$, so $\forall i \notin \mathcal{I}$ it holds that
        \begin{align*}
        \theta^{k} 
        %\geq \tilde{\mu}\|\bm{\tilde{X}}^{k}-\bm{\tilde{X}}^{\natural}\|_{2,1}+C_{W}\| \bm{\tilde{Z}} \|_{F} 
        \geq \big\| - (\bm{W}^k[:,i])^T \bm{\tilde{S}} (  \bm{\tilde{X}}^k -   \bm{\tilde{X}}^{\natural}) +  (\bm{W}^k[:,i])^T \bm{\tilde{Z}} )  \big\|_2, 
        \end{align*}
        which implies $\| \bm{\tilde{X}}^{k+1}[i,:]  \|_2=0, \forall i \notin \mathcal{I}$ by the definition of $\eta_{\theta^{k}} .$ Therefore, by induction, we have
        $\| \bm{\tilde{X}}^k[i,:]  \|_2=0, \forall i \notin \mathcal{I}, \forall k.$
\end{proof}     
\section{Proof of Theorem 3}\label{pr:th3}
\begin{proof}
        Similar to the proof of Theorem \ref{thm:convergence1}, we first show that there are no false positive in $ \bm{\tilde{X}}^k$. % $ \bm{\tilde{X}}^k[i,:]  = \bm{0}, \forall i \notin \mathcal{I}, \forall k$.    
        Take arbitrary $(\bm{\tilde X}^{\natural}, \bm{\tilde Z}) \in \mathcal{X}(\beta,s,0)$ and let $\mathcal{I} = \supp( \psi( \bm{\tilde X}^{\natural}) )$. We prove the no-false-positive property by induction.
        
        (i) When $k = 0$, it is satisfied since $\bm{\tilde{X}}^{0}=\bm{0}$.
        
        (ii) Fix $k$ and suppose that 
        $ \bm{\tilde{X}}^k[i,:]  = \bm{0}$ for all $i \notin \mathcal{I}$.
    Then we have
        \begin{align*}
                \bm{\tilde{X}}^{k+1}[i,:] 
                &= \eta_{\theta^k} \big( \bm{\tilde{X}}^k[i,:]  + \gamma^k(\bm{W}[:,i])^T (\bm{\tilde{Y}}   -\bm{\tilde{S}}  \bm{\tilde{X}}^k)  \big)\\
                &= \eta_{\theta^k} \big( - \gamma^k(\bm{W}[:,i])^T\bm{\tilde{S}} (\bm{\tilde{X}}^k   - \bm{\tilde{X}}^{\natural})  \big)
        \end{align*}
        for all $i \notin \mathcal{I} $.
        As the thresholds are taken as  
        $\theta^{k}=  \tilde{\mu} \gamma^k \sup_{\bm{\tilde{X}}^{\natural}}   \| \bm{\tilde{X}}^{k}-\bm{\tilde{X}}^{\natural} \|_{2,1}  $ and $\bm{W} \in \mathcal{W}(\bm{\tilde{S}})$, it holds that
        \begin{align*}
        \theta^{k}
        \geq \tilde{\mu} \gamma^k \big\|  \bm{\tilde{X}}^{k}-\bm{\tilde{X}}^{\natural} \big\|_{2,1} 
        %&= \tilde{\mu} \gamma^k \sum_{j \in \mathcal{I}} \big\| \bm{\tilde{X}}^k[j,:] - \bm{\tilde{X}}^{\natural}[j,:] \big\|_2  \\ 
        %&\geq \sum_{j \in \mathcal{I}} \gamma^k \big| (\bm{W}[:,i])^T \bm{\tilde{S}}[:,j] \big| \big\| \bm{\tilde{X}}^k[j,:]   - \bm{\tilde{X}}^{\natural}[j,:] \big\|_2 \\
        \geq \big\|  - \gamma^k(\bm{W}[:,i])^T\bm{\tilde{S}} (\bm{\tilde{X}}^k   - \bm{\tilde{X}}^{\natural}) \big\|_2, \\
         \forall i \notin \mathcal{I},
        \end{align*}
        which implies $\| \bm{\tilde{X}}^{k+1}[i,:]  \|_2=0, \forall i \notin \mathcal{I}$ by the definition. Thus,
        $ \bm{\tilde{X}}^k[i,:] =\bm{0},\forall i \notin \mathcal{I}, \forall k.$
        
        In the next step, we consider the components on $\mathcal{I}$. For all $i \in \mathcal{I}$, we have
        \begin{align}
        &\bm{\tilde X}^{k+1}[i,:] 
        %&= \eta_{\theta^k} \Big( \bm{\tilde{X}}^k[i,:]  - \gamma^k (\bm{W}[:,i])^T \bm{\tilde{S}}[:,\mathcal{I}] (  \bm{\tilde{X}}^k[\mathcal{I},:] -   \bm{\tilde{X}}^{\natural}[\mathcal{I},:])   \Big) \nonumber \\
        \in \underbrace{\bm{\tilde{X}}^k[i,:]  - \gamma^k(\bm{W}[:,i])^T \bm{\tilde{S}}[:,\mathcal{I}] \big(  \bm{\tilde{X}}^k[\mathcal{I},:] -   \bm{\tilde{X}}^{\natural}[\mathcal{I},:] \big) }_{(T1)} \notag\\
        & - \theta^k \partial \Vert\bm{\tilde X}^{k+1}[i,:]\Vert_{2}, \label{eq:thm3-1}
        \end{align}
        where $\partial \Vert\bm{\tilde X}^{k+1}[i,:]\Vert_{2}$ is defined in \eqref{l2norm}. As we choose $\bm{W} \in  \mathcal{W}(\bm{\tilde{S}})$, so  $(\bm{W}[:,i])^T \bm{\tilde{S}}[:,i] = 1$, then $(T1)$ can be expressed as 
        \begin{align*}
        %&\bm{\tilde{X}}^k[i,:]  - \gamma^k(\bm{W}[:,i])^T \bm{\tilde{S}}[:,\mathcal{I}] \Big(\bm{\tilde{X}}^k[\mathcal{I},:]  - \bm{\tilde{X}}^{\natural}[\mathcal{I},:] \Big) \\
        %&= \bm{\tilde{X}}^k[i,:] - \gamma^k \sum_{j\in \mathcal{I}, j \ne i} (\bm{W}[:,i])^T \bm{\tilde{S}}[:,j] \Big(\bm{\tilde{X}}^k[j,:]  - \bm{\tilde{X}}^{\natural}[j,:] \Big)  \\
        %& \qquad - \gamma^k (\bm{\tilde{X}}^k[i,:]  - \bm{\tilde{X}}^{\natural}[i,:]) \\
   \bm{\tilde{X}}^{\natural}[i,:] \underbrace{ - \gamma^k \sum_{j\in \mathcal{I}, j \ne i} (\bm{W}[:,i])^T \bm{\tilde{S}}[:,j](\bm{\tilde{X}}^k[j,:]  - \bm{\tilde{X}}^{\natural}[j,:])} \\
   \underbrace{ + (1 - \gamma^k) (\bm{\tilde{X}}^k[i,:]  - \bm{\tilde{X}}^{\natural}[i,:])}_{(T2)}.
        \end{align*}
        Hence, \eqref{eq:thm3-1} can be rewritten as
        \begin{align*}
        \bm{\tilde X}^{k+1}[i,:] - \bm{\tilde{X}}^{\natural}[i,:] \in  (T2) - \theta^k \partial \ell_{2,1}(\bm{\tilde X}^{k+1}[i,:]) 
        \end{align*}
        Then one can get that for all $i \in \mathcal{I}$
        \begin{align*}
        \Vert \bm{\tilde X}^{k+1}[i,:] - \bm{\tilde{X}}^{\natural}[i,:]  \Vert_2  
        %& \leq \sum_{j\in \mathcal{I}, j \ne i} \gamma^k|(\bm{W}[:,i])^T \bm{\tilde{S}}[:,j] | \Vert(\bm{\tilde{X}}^k[j,:]  - \bm{\tilde{X}}^{\natural}[j,:]) \Vert_2 \\  
        %& \qquad + \theta^k + |1 - \gamma^k|\| \bm{\tilde{X}}^k[i,:]  - \bm{\tilde{X}}^{\natural}[i,:]\|_2 \\
        \leq \tilde{\mu} \gamma^k \sum_{j\in \mathcal{I}, j \ne i} \Vert(\bm{\tilde{X}}^k[j,:]  - \bm{\tilde{X}}^{\natural}[j,:]) \Vert_2 \\
         + \theta^k + |1 - \gamma^k|\| \bm{\tilde{X}}^k[i,:]  - \bm{\tilde{X}}^{\natural}[i,:]\|_2 .
        \end{align*}
        By no-false-positive property, we have
        \begin{align*}
        \Vert \bm{\tilde{X}}^{k+1} - \bm{\tilde{X}}^{\natural} \Vert_{2,1} %&= \sum_{i \in \mathcal{I}} \Vert \bm{\tilde{X}}^{k+1}[i,:] -   \bm{\tilde{X}}^{\natural}[i,:] \Vert_{2} \\
        %&\leq \sum_{i \in \mathcal{I}}  \Big( \tilde{\mu}\gamma^k \sum_{j\in \mathcal{I}, j \ne i} \big\Vert(\bm{\tilde{X}}^k[j,:]  - \bm{\tilde{X}}^{\natural}[j,:]) \big\Vert_2 + \theta^k  +|1 - \gamma^k|\| \bm{\tilde{X}}^k[i,:]  - \bm{\tilde{X}}^{\natural}[i,:]\|_2   \Big) \\ 
        %&=(|\mathcal{I}| - 1) \tilde{\mu}\gamma^k \sum_{i \in \mathcal{I}} \big\Vert(\bm{\tilde{X}}^k[i,:]  - \bm{\tilde{X}}^{\natural}[i,:]) \big\Vert_2 + |\mathcal{I}|\theta^k+ |1 - \gamma^k|\| \bm{\tilde{X}}^k  - \bm{\tilde{X}}^{\natural}\|_{2,1}  \\
        \leq \tilde{\mu}\gamma^k(|\mathcal{I}| - 1) \Vert \bm{\tilde{X}}^k - \bm{\tilde{X}}^{\natural} \Vert_{2,1}  + |\mathcal{I}|\theta^k \\+ |1 - \gamma^k|\| \bm{\tilde{X}}^k  - \bm{\tilde{X}}^{\natural}\|_{2,1}.
        \end{align*}
        
        Finally, we take supremum over $(\bm{\tilde{X}}^{\natural},\bm{0})\in\mathcal{X}(\beta, s, 0)$ with $|\mathcal{I}|\leq s$, we have
        \begin{align*}
        \sup _{\bm{\tilde X}^{\natural}}  \Vert \bm{\tilde{X}}^{k+1} - \bm{\tilde{X}}^{\natural} \Vert_{2,1}  \leq  \big( \tilde{\mu}\gamma^k(s - 1) + |1 - \gamma^k|\big)\\ \sup _{\bm{\tilde X}^{\natural}}  \Vert  \bm{\tilde{X}}^k - \bm{\tilde{X}}^{\natural} \Vert_{2,1}   + s\theta^k.
        \end{align*}
        Taking $c^{\tau} = -\log\big( \gamma^\tau(2 \tilde{\mu} s-\tilde{\mu}) + |1 - \gamma^\tau |\big)$, then from the fact that $\theta^{k}=\tilde{\mu} \gamma^k \sup _{\bm{\tilde{X}}^{\natural}} \|\bm{\tilde{X}}^{k}-\bm{\tilde{X}}^{\natural}\|_{2,1} $ in \eqref{eq:thm3-condition2} we obtain
        \begin{align*}
        \sup _{\bm{\tilde X}^{\natural}} \Vert \bm{\tilde{X}}^{k+1} - \bm{\tilde{X}}^{\natural} \Vert_{2,1}    
        %&\leq \big( \gamma^k(2 \tilde{\mu} s-\tilde{\mu}) + |1 - \gamma^k|\big) \sup _{\bm{\tilde X}^{\natural}}  \Vert \bm{\tilde{X}}^k - \bm{\tilde{X}}^{\natural} \Vert_{2,1} \\
        %& \quad \vdots \\
        &\leq \exp\Big(-\sum_{\tau=0}^{k}c^{\tau} \Big) \sup _{\bm{\tilde X}^{\natural}}  \Vert \bm{\tilde{X}}^0 - \bm{\tilde{X}}^{\natural} \Vert_{2,1}  \\
        &\leq s \beta \exp\Big(-\sum_{\tau=0}^{k}c^{\tau} \Big).
        \end{align*}
        Hence, we get the following upper bound with respect to the Frobenius norm:
        \begin{align}
        \sup _{\bm{\tilde X}^{\natural}} \Vert \bm{\tilde{X}}^{k+1} - \bm{\tilde{X}}^{\natural} \Vert_{F}  \notag 
        \leq \sup _{\bm{\tilde X}^{\natural}} \Vert \bm{\tilde{X}}^{k+1} - \bm{\tilde{X}}^{\natural} \Vert_{2,1} \notag \\
         \leq s \beta \exp\Big(-\sum_{\tau=0}^{k}c^{\tau} \Big).
        \end{align}
        Therefore, the error bound holds uniformly
        for all $(\bm{\tilde X}^{\natural}, \bm{0}) \in \mathcal{X}(\beta, s, 0)$.
        
        Lastly, we verify that $c^\tau >0$ for all $\tau$.
        The assumption $s<(1+1 / \tilde{\mu}) / 2$ gives $2\tilde{\mu}s-\tilde{\mu}<1$.
        If $0<\gamma^\tau \leq 1$, then $
 \gamma^\tau(2 \tilde{\mu} s-\tilde{\mu}) + |1 - \gamma^\tau | 
 = \gamma^\tau(2 \tilde{\mu} s-\tilde{\mu}-1)  + 1 <1.
        $
        If $1<\gamma^\tau < 2/(1+2\tilde{\mu} s - \tilde{\mu})$, then
        $
 \gamma^\tau(2 \tilde{\mu} s-\tilde{\mu}) + |1 - \gamma^\tau | 
 = \gamma^\tau(2 \tilde{\mu} s-\tilde{\mu}+1)  - 1 <1.
        $
        Thus $c^\tau >0$ for all $\tau$.
        
\end{proof}                              
%Remark : This proof shows that every matrix under the thresholding function that we choose will make any initial matrix converge to the target matrix under the ALISTA-GS network that is trained by data matrices.

\end{appendices}

\bibliographystyle{ieeetr}
\bibliography{refs}

\end{document}